%% file: main.tex
\documentclass[10pt,journal,compsoc]{IEEEtran}
\usepackage{booktabs} % For formal tables
\usepackage{acronym}
\usepackage{framed}
\usepackage{subfigure}
\usepackage{multirow}
\usepackage{tabularx}
\usepackage{algpseudocode}
\usepackage{float}
\usepackage[pdftex]{graphicx}

\ifCLASSOPTIONcompsoc
  \usepackage[nocompress]{cite}
\else
  \usepackage{cite}
\fi
\usepackage{amsmath}
\usepackage{amsthm}
\usepackage{url}
\begin{document}
\title{On the costs and profit of software defect prediction}

\author{Steffen~Herbold
\IEEEcompsocitemizethanks{\IEEEcompsocthanksitem S. Herbold is with the
University of Goettingen, Institute of Computer Science, 
Goettingen, Germany.\protect\\
E-mail: herbold@cs.uni-goettingen.de}% <-this
\thanks{}}


\markboth{IEEE Transactions on Software Engineering,~Vol.~X, No.~X,
MONTH YEAR}%
{Herbold: On the costs and profit of software defect prediction}



\input{commands}
\input{acronyms}


\IEEEtitleabstractindextext{%
\input{abstract}

}

\maketitle

\IEEEdisplaynontitleabstractindextext

% \fi
\IEEEpeerreviewmaketitle

\acresetall
\input{introduction}
\input{relatedwork}

\input{generalmodel}
\input{conditions}
\input{init}
\input{experiments}

\input{discussion}
\input{conclusion}

\section*{Acknowledgements}

This work is partially funded by DFG Grant 402774445.

%\appendix
%\input{response_to_reviews}
%\input{appendix}

% or


\ifCLASSOPTIONcaptionsoff
  \newpage
\fi

\bibliographystyle{IEEEtran}
\bibliography{literature}

\begin{IEEEbiography}[{\includegraphics[width=1in,height=1.25in,clip,keepaspectratio]{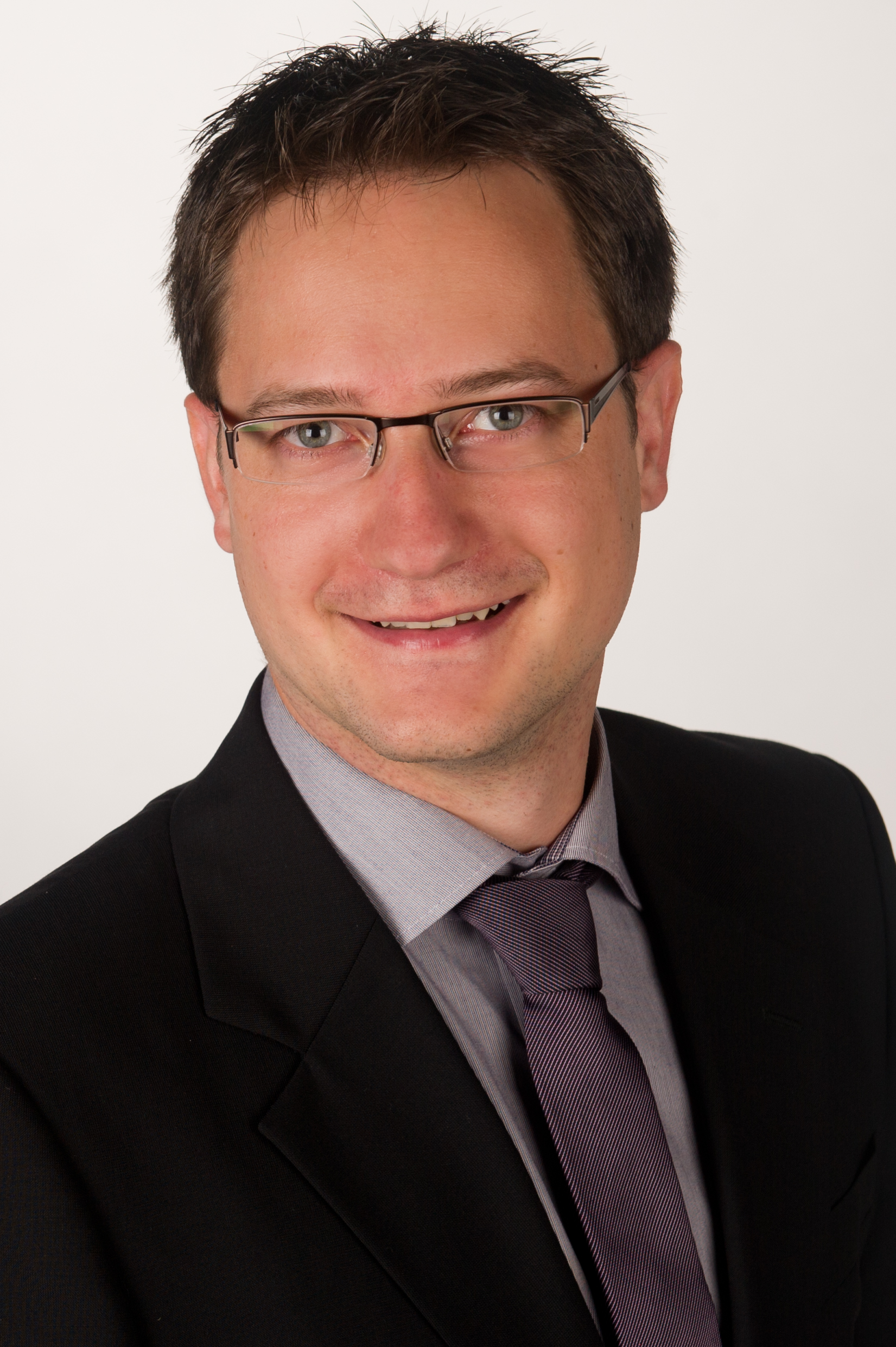}}]{Steffen
Herbold} PD Dr. Steffen Herbold currently manages a Machine Learning professorship at the
Institute of Computer Science of the Georg-August-Universit\"{a}t
G\"{o}ttingen. His research is focused on the application of data science
methods and their applications in software engineering as well as software
engineering to data science methods.
\end{IEEEbiography}

\end{document}

%% file: commands.tex
\newcommand{\etal}{~\textit{et al.}}
\newcommand{\AUC}{\textit{AUC}}
\newcommand{\FMEAS}{\textit{F-measure}}
\newcommand{\GMEAS}{\textit{G-measure}}
\newcommand{\MCC}{\textit{MCC}}
\newcommand{\AUCEC}{\textit{AUCEC}}
\newcommand{\RELB}{$RelB_{20\%}$}
\newcommand{\NECM}{$NECM_{15}$}
\newcommand{\RANKSCORE}{\textit{rankscore}}
\newcommand{\RANKSCORES}{\textit{rankscores}}
\newcommand{\RECALL}{\textit{recall}}
\newcommand{\PRECISION}{\textit{precision}}
\newcommand{\ERROR}{\textit{error}}

\newtheorem{definition}{Definition}%[section]
\newtheorem{notation}{Notation}%[section]
\newtheorem{theorem}{Theorem}%[section]
\newtheorem{lemma}{Lemma}%[section]
\newtheorem{corollary}{Corollary}

%% file: acronyms.tex
\acrodef{ANOVA}{ANalysis Of VAriance}
\acrodef{AST}{Abstract Syntax Tree}
\acrodef{AUC}{Area Under the ROC Curve}
\acrodef{Ca}{Afferent Coupling}
\acrodef{CBO}{Coupling Between Objects}
\acrodef{CCA}{Canonical Correlation Analysis}
\acrodef{CD}{Critical Distance}
\acrodef{Ce}{Efferent Coupling}
\acrodef{CFS}{Correlation-based Feature Subset}
\acrodef{CLA}{Clustering and LAbeling}
\acrodef{CODEP}{COmbined DEfect Predictor}
\acrodef{CPDP}{Cross-Project Defect Prediction}
\acrodef{DBSCAN}{Density-Based Spatial Clustering}
\acrodef{DCV}{Dataset Characteristic Vector}
\acrodef{DTB}{Double Transfer Boosting}
\acrodef{fn}{false negative}
\acrodef{fp}{false positive}
\acrodef{GB}{GigaByte}
\acrodef{HL}{Hosmer-Lemeshow}
\acrodef{ITS}{Issue Tracking System}
\acrodef{JIT}{Just In Time}
\acrodef{LCOM}{Lack of COhession between Methods}
\acrodef{LOC}{Lines Of Code}
\acrodef{MDP}{Metrics Data Program}
\acrodef{MI}{Metric and Instances selection}
\acrodef{MODEP}{MultiObjective DEfect Predictor}
\acrodef{MPDP}{Mixed-Project Defect Prediction}
\acrodef{NN}{Nearest Neighbor}
\acrodef{PCA}{Principle Component Analysis}
\acrodef{RAM}{Random Access Memory}
\acrodef{RFC}{Response For a Class}
\acrodef{SCM}{SourceCode Management system}
\acrodef{SVM}{Support Vector Machine}
\acrodef{TCA}{Transfer Component Analysis}
\acrodef{tn}{true negative}
\acrodef{tp}{true positive}
\acrodef{RBF}{Radial Basis Function}
\acrodef{ROC}{Receiver Operating Characteristic}
\acrodef{UMR}{Unified Metric Representation}
\acrodef{VCB}{Value-Cognitive Boosting}
\acrodef{WPDP}{Within-Project Defect Prediction}
\acrodef{QA}{Quality Assurance}

%% file: abstract.tex
\begin{abstract}
Defect prediction can be a powerful tool to guide the use of quality assurance
resources. However, while lots of research covered methods for defect
prediction as well as methodological aspects of defect prediction research, the
actual cost saving potential of defect prediction is still unclear. Within this
article, we close this research gap and formulate a cost model for software
defect prediction. We derive mathematically provable boundary conditions that must
be fulfilled by defect prediction models such that there is a positive
profit when the defect prediction model is used. Our cost model includes aspects
like the costs for quality assurance, the costs of post-release defects, the possibility
that quality assurance fails to reveal predicted defects, and the relationship
between software artifacts and defects. We initialize the cost model using
different assumptions, perform experiments to show trends of the behavior of
costs on real projects. Our results show that the unrealistic assumption that
defects only affect a single software artifact, which is a standard practice in
the defect prediction literature, leads to inaccurate cost estimations.
Moreover, the results indicate that thresholds for machine learning metrics
are also not suited to define success criteria for software defect prediction. 
\end{abstract}

%% file: introduction.tex
\IEEEraisesectionheading{\section{Introduction}}
\label{sec:introduction}

\IEEEPARstart{R}{esearch} regarding software defect prediction for the accurate
prediction of post-release defects in software is an ongoing and still
unresolved research topic, that was already discussed in hundreds of
publications \cite{Catal2009a,Hall2012,Hosseini2017a}. Current research focuses
on problems like cross-project defect prediction~(e.g., \cite{Herbold2017b}),
heterogeneous defect prediction~(e.g., \cite{Nam2017, Jing2015}), unsupervised
defect prediction~(e.g., \cite{Nam2015b,Zhang2016}), and just-in-time defect
prediction (e.g., \cite{Huang2017}). Additionally, researchers have turned their
attention to how defect prediction research should be conducted, e.g., reducing
the bias through sampling approaches~\cite{Tantithamthavorn2017}, the impact of
hyper parameter tuning~\cite{Tantithamthavorn2016}, suitable baseline
comparisons~\cite{Krishna2018} or general guidelines that should be
considered~\cite{Tantithamthavorn2018}. While all of the above contribute to the
advancement of the defect prediction state of the art, there are also multiple
publications that question the progress of the state of the art through
replications in recent years, as they demonstrate that older
(e.g.,~\cite{Herbold2017b}) or trivial (e.g.,~\cite{Yang2016, Zhou2018})
approaches are comparable too or even better than more complex recent approaches
from the state of the art.

The problems with replications of prior results lead to the question, if defect prediction can currently really help
developers and organizations to reduce costs. While this aspect
is crucial, there are only few publications that cover cost models for defect
prediction (see Section~\ref{sec:relatedwork}). Moreover, the existing cost
models have several limitations, e.g., a 1-to-1 relationship between software artifacts and defects, or the assumption
that quality assurance is perfect, i.e., all predicted defects are found.
Additionally, researchers often use of ``standard'' machine learning
metrics like \textit{precision}, \textit{recall}, \textit{F-Measure}, \textit{MCC}, \textit{AUC} and others instead of a cost
model~\cite{Hosseini2017a, Herbold2017}.
While these metrics are suitable to estimate the general performance of defect
prediction models, they are not suitable to answer the question if defect
prediction can actually save costs, i.e., have a positive profit.

With this article, we want to close this research gap through the specification
of a general cost model for defect prediction. Our cost model takes the costs
for quality assurance, the costs of defects, the relationship between defects
and software artifacts, the possibility that quality assurance may fail to
reveal defects as well as one-time and continuous costs for the execution of
the defect prediction into account. 

% We derive and mathematically proof
% conditions that must be fulfilled for defect prediction to be cost saving. We then proceed to show how the general cost model can be initialized
% using different assumptions. Through simulations, we show that it is
% quite possible that defect prediction models cannot reduce costs if
% their quality is not sufficently high. Additionally, our simulations show that
% the standard assumption that one artifact may be defective or not (binary
% labeling) and defect counts for artifacts (each defect still may only belong to
% one artifact), leads to very different cost considerations than a more realistic
% $n$-to-$m$ relationship between artifacts and defects, i.e., each defect may
% affect multiple artifacts and each artifact may contain multiple defects. 

The contributions of this article are the following: 
\begin{itemize}
  \item A general cost model for software defect prediction.
  \item Mathematically proven boundary conditions on cost saving defect
  predictions.
  \item Initializations of the cost model with realistic assumptions that can be
  used by researchers and practitioners for the evaluation of defect prediction
  models, including guidelines on how to use the cost model.
  \item Through the work on the cost model, we discovered a principle problem
  in current defect prediction papers, i.e., that we do not account for the
  fact that there is an $n$-to-$m$ relationship between software artifacts and
  defects. Through simulations of defect prediction models on real-world data,
  we have shown that the results, especially in terms of costs, may change if
  this relationship is considered.
\end{itemize}

The remainder of this article is structured as follows. In
Section~\ref{sec:relatedwork} we discuss existing cost models and cost-sensitive
metrics for software defect prediction. Then, we formally specify the problem of
software defect prediction in Section~\ref{sec:formal-specification} and derive
the general cost model for defect prediction from this specification in
Section~\ref{sec:general-cost-model}. In Section~\ref{sec:conditions} we proof
properties that must be fulfilled by defect prediction models in order to have a
positive profit. Then, we show how the general cost model can be initialized
under different assumptions in Section~\ref{sec:init}. Through simulation
experiments, we evaluate trends of the boundary conditions for cost saving
defect predictions, as well as the impact
of different assumptions on the initialization of the cost model in
Section~\ref{sec:experiments}. We proceed with a discussion of our cost model,
including guidelines on how to use our model and threats to the validity of our
work in Section~\ref{sec:discussion} and conclude in
Section~\ref{sec:conclusion}.

%% file: relatedwork.tex
\section{Related Work}
\label{sec:relatedwork}

Most of the defect prediction literature does not use cost-sensitive
evaluations, but standard machine learning measures based on the confusion
matrix, e.g., \textit{precision}, \textit{recall}, \textit{F-Measure}, and
\textit{MCC}. Within this section, we discuss approaches for the
evaluation of defect prediction models that directly take costs into account. We differentiate
between cost metrics and cost models. A cost metric takes specific parts of the
costs into account, but does not try to actually estimate the complete costs
related to the defect prediction model. A cost model combines multiple or all
relevant aspects associated with the costs of a defect prediction model. As a
consequence, cost metrics are only indicators of cost effectiveness, whereas
cost models can be used to calculate the costs as well as the profit of
defect prediction models. After the discussion of costs models for defect
prediction, we present a broader overview about related work on cost modeling,
both in software engineering as well as other domains.

\subsection{Cost Metrics}

There are also multiple performance metrics in the literature, which take costs
into account. Ohlson and Alberg~\cite{Ohlsson1996}, as well as
Rahman\etal~\cite{Rahman2012} defined variants of \ac{ROC} curves that take costs into account. These
\ac{ROC} curves are defined over the number of defects found versus the number
of artifacts that have to be considered. Using the area under this \ac{ROC}
curve, a threshold independent and cost sensitive performance measure is
defined. The variant by Ohlson and Alborg is also
known as lift chart~\cite{Jiang2008}. Rahman\etal~\cite{Rahman2012} also
considered using the percentage of lines of code instead of the number of
artifacts. However, the results are similar. Hemmati\etal~\cite{Hemmati2015} use
a similar variant of \ac{ROC} that considers the percentage of defects detected
versus the percentage of lines of code. Arisholm and Briand~\cite{Arisholm2006}
propose an approach similar to a ROC curve. They plot the percentage of
defects found, as well as the percentage of code considerd both on the y-axis
versus different cutoff values for a prediction model on the x-axis. The area
between these lines indicates the cost saving potential.

Canfora\etal~\cite{Canfora2013} use costs defined by the lines of code
that are predicted as defective as a criterion for a multi-objective
optimization algorithm for cross-project defect prediction. Another measure that is
considered is the number of modules that must be visited before 80\% of the
defects are found (e.g., \cite{Jureczko2010}). Similarly, some authors
considered the number of defects found if the top twenty percent of the
predictions are considered (e.g., \cite{Zhang2015a}), i.e., the predictions
with the highest scores.

% While the cost metrics cannot be directly be used as a cost model, they
% contain valuable information relevant for our cost model in how they model
% costs. Multiple of these cost metrics use of the lines of code for the
% estimation of quality assurance costs. We adopt this for a size-aware
% initilization of our cost model.

\subsection{Cost Models}

A cost model with similar traits to our work was proposed by
Khoshgoftaar and Allen~\cite{Khoshgoftaar1998}. In their work, the authors
observe that the costs for false positives and false negatives are different.
They model the expected costs of misclassifications through two
constants that represent the costs of unnecessary quality assurance in case of false positives and the costs
of missed defects in case of false negatives. Later, this approach was extended
to consider these costs as a ratio~\cite{Khoshgoftaar2004}. Drummond and
Holte~\cite{Drummond2006} propose the use of cost curves for classifier comparisons.
The cost curves consider the same cost model as Khoshgoftaar and
Allen~\cite{Khoshgoftaar1998}, i.e. the expected costs of misclassifications.
However, instead of assuming a single constant, they propose to use a \ac{ROC}
curve of the expected costs versus different cost ratios.

Another cost model for defect prediction was proposed by Zhang and
Cheung~\cite{Zhang2013}. The model is also similar to the work by Khoshgoftaar
and Allen~\cite{Khoshgoftaar1998}. However, they also took the costs for true
positives into account, i.e., the costs of necessary quality assurance that can prevent post release defects.
Based on this model, the authors derive a criterion for cost-effectiveness of
defect prediction that must be fulfilled such that defect prediction outperforms
randomly selecting artifacts for quality assurance as well as applying quality
assurance to all artifacts.

While our cost model has the same general structure as the models by
Khoshgoftaar et al.~\cite{Khoshgoftaar1998, Khoshgoftaar2004} and Zhang and
Cheung~\cite{Zhang2013}, both may only be considered as a special case of our
approach as they have several limitations which our model overcomes.
They assume 1-to-1 relationships between software artifacts and defects, i.e.,
binary defect labels for the artifacts. This is unrealistic as software
artifacts may contain multiple defects and defects may affect multiple software
artifacts. Moreover, both models do not take into account that quality
assurance is not perfect and may not be able to detect predicted defects. Additionally,
the models assume constant costs for both quality assurance as well as defects.
Our general model allows individual costs for each artifact and each defect and
we cover the constant costs as a special case we consider. Finally, both
cost models ignore one-time and continuous costs that occur if a defect prediction model is
used within a development process. The cost curves by Drummond and
Holte~\cite{Drummond2006} could also be used with our cost model, to visualize the cost
savings for different cost ratios.

\subsection{Cost Model for Other Applications}

While this article is focused on cost models for defect prediction, we also want
to give a brief overview on cost models for other applications. In software engineering,
similar cost models to our work were proposed for the reliability assessment of
software with the goal to determine the time costs of releases~\cite{Pham2003}.
Such models estimate the costs based on the costs for quality assurance and for fixing
defects prior to a release in comparison to the costs due to fixing defects
after the release and the costs for delaying a release. The reliability models use
stochastic processes to model the expected number of defects in the
software, e.g., a non-homogeneous Poisson process~\cite{Kingman1992}. Such a
stochastic process is not required for our cost model because the defects are
known from empirical data. Similar to our work, the cost models for software
reliability use constants for the costs of quality assurance and due to
defects. The only cost factor that is not assumed as constant are the penalty
costs for delaying a release, which are modelled as a function that
monotonically increases with the delay.

Stolfo\etal~\cite{Stolfo2000} created a cost model for the evaluation of machine
learning based fraud detection. While the domain is different, the goal is
similar to our work for defect prediction, i.e., to provide means beyond
standard machine learning metrics to assess the impact of a prediction approach
on costs. The assumptions behind the structure of the cost model are similar to
defect prediction cost models. The authors compute the cost savings based on the
effort spent due to predictions, costs saved due to true positive predictions of
fraud, and false negative misses of fraud. The costs for effort spent is similar
to the quality assurance costs in defect prediction cost models and assumed by
Stolfo\etal~\cite{Stolfo2000} to be constant. The costs for detecting or missing
fraud is similar to the costs of defects. A big advantage of the fraud detection
cost model over our initializations of the cost model is that the authors know
the loss due to fraud, because this is the amount of money in a fraudulent
transaction. In comparison, we have to rely on a constant that represents the
mean costs per defect.

In general, the literature on cost models follows a pattern for the creation of
cost models similar to our work: the authors determine factors related to the
costs from the literature and/or experience and create a cost model as the sum
of these cost factors. For example, Patry\etal~\cite{Patry2015} use this
approach to assess the cost of lithium-ion battery cells, Etkin~\cite{Etkin2004}
assess different factors of costs associated with oil spillages, and
Pugliatti\etal~\cite{Pugliatti2007} for the cost of epilepsy in Europe. We note
that the complexity of the cost models is also impacted by the amount of
research invested into understanding different cost factors in detail. For
example, decades of research on economics have led to complex cost models
that can model whole economies by describing different actors through
well-understood stochastic processes. An example for such a model is the work by
Nakumura and Steinsson~\cite{Nakamura2010}: the authors created a cost model
that takes household consumptions of different goods over time as well as the
labor and production behavior of companies into account to analyze the impact of
economic shocks on money non-neutrality.

%% file: generalmodel.tex
\section{Specification of defect prediction}
\label{sec:formal-specification}

Defect prediction models are used to assess the risk of software artifacts and
guide quality assurance efforts in order to prevent post-release defects. Formally, 
let $S$ be a software product that consists of artifacts $s \in S$. These
artifacts may be modules, files, classes or methods. The software product
contains defects $d \in D$, whereas each defect belongs to one or more
artifacts.
We denote the artifacts that $d$ belongs to as a set and define $d =
\{s^d_1,, \ldots, s^d_n\}$ to denote that the artifacts $s^d_1, \ldots s^d_n \in S$ are
defective because of defect $d$. Vice versa, we denote the defects that affect
an artifact $s$ as $d(s) = \{d \in D: s \in d\}$. Since one defect can belong to
multiple artifacts, and artifacts can be affected by multiple defects, we have
an $n$ to $|d(s)|=m$ relationship between artifacts and defects. 
Given the defects $D$, we can divide the artifacts $S$ into two disjunctive
sets $S_{DEF}$ and $S_{CLEAN}$ such that $S_{DEF}$ contains all artifacts that
contain defects and $S_{CLEAN}$ contains all artifacts without any defects,
i.e., 
\begin{equation}
\begin{split}
S_{DEF} &= \bigcup_{d \in D} d \\
S_{CLEAN} &= S \setminus S_{DEF}.
\end{split}
\end{equation}

The goal of defect prediction models is to estimate $S_{DEF}$ and
$S_{CLEAN}$. Thus, a defect prediction model is a function
\begin{equation}
\label{eq:defect-prediction-model}
h: S \rightarrow \{Defective, Clean\}.
\end{equation}
In the following, we use the labels $1 = Defective$ and $0 = Clean$, i.e., the
prediction model classifies artifacts into clean and defective artifacts. With
the exception of a publication by Hemmati\etal~\cite{Hemmati2015}, the current
state of the art assumes each artifact that is correctly labeled as defective,
predicts all defects in that affect an artifact correctly. However, this is not
necessarily the case, because a defect $d \in D$ may affect multiple artifacts.
A defect $d$ is successfully predicted by a defect prediction model if all
artifacts $s \in d$ are labeled as defective, i.e., $h(s) = 1$ for all $s \in
d$. We define the set of
defects that are successfully predicted by a defect prediction model $h$ as
\begin{equation}
D_{PRED} = \{d \in D: \forall s \in d~|~h(s) = 1\}
\end{equation}
and the set of defects that are missed as
\begin{equation}
D_{MISS} = \{d \in D: \exists s \in d~|~h(s)=0\} = D \setminus D_{PRED}
\end{equation}

To get a better understanding of our definitions, consider an example with
three software artifacts $S = \{s_1, \ldots, s_3\}$ and two defects $D =
\{d_1=(s_1), d_2=(s_1, s_2)\}$. Thus, $s_1$ is affected by both defects, $s_2$ is only
affected by defect $d_2$ and $s_3$ is clean. Let us consider a defect
prediction model that predicts $s_1$ as defective and the other artifacts as
clean, i.e., $h(s_1)=1$, $h(s_2)=0$, and $h(s_3)=0$. This means that
$D_{PRED}=\{d_1\}$ because all artifacts that $d_1$ affects are predicted as
defective by $h$ and $D_{MISS}=\{d_2\}$ because the artifact $s_2$ that is
affected by $d_2$ is predicted as clean. 

Depending on the prediction of the model and the actual post-release defects
that are observed in the software, there are four possible outcomes of a
prediction.
\begin{enumerate}
  \item The defect prediction model predicts a post-release defect in an
  artifact correctly. This is called a true postive. We denote the
  number of true positives as $tp$.
  \item The defect prediction model falsely predicts a post-release defect in an
  artifact. This is called a false positive. We denote the number of false
  positives as $fp$.
  \item The defect prediction model correctly predicts that an artifact does not
  contain a post-release defect. This is called true negative. We denote
  the number of true negatives as $tn$.
  \item The defect prediction model falsely predicts that an artifact does not
  contain a post-release defect. This is called a false negative. We denote the
  number of false negatives as $fn$.
\end{enumerate}

\section{General Cost Model}
\label{sec:general-cost-model}

The use of defect prediction models affects several costs in a development
process. In general, we must account for the following factors. 
 
\begin{itemize}
  \item $cost_{INIT}$, i.e., one-time costs for the introduction
  of defect prediction into a development process. 
  \item $cost_{EXEC}$, i.e., continuous costs related to the usage of the defect
  prediction model in the development process, e.g, for the
  preparation of data, analysis of prediction results, or the re-training of
  models. 
  \item $cost_{QA}$, i.e., costs due to additional quality assurance
  measures that are applied as a result of the predictions made by the model.
  \item $cost_{DEF}$, i.e., the costs due to post-release defects. These costs
  include not only the directly incurring costs, e.g., due to a loss in revenue or
  contract penalties, but also the maintenance costs for deploying patches in
  the wild, the costs of regression testing, or costs due to an increased effort
  for the correction due to restrictions on the allowed changes to the source
  code after the initial release.
\end{itemize}

The costs for actually fixing the defects in the artifacts are
not a relevant factor for the costs of defect prediction. These costs either
occur as a result of quality assurance before the release, or due to a
post-release defect after the release. Thus, the costs for fixing a defect will
always be present and cannot be changed due to the defect prediction, only
losses due to post-release defects can be prevented. 

If all the costs are known, we can calculate the costs associated with acting
according to a defect prediction model as
\begin{equation}
\label{eq:cost1}
cost = cost_{INIT}+cost_{EXEC}+cost_{QA}+cost_{DEF}.
\end{equation} 

\subsection{One-time costs and continuous costs}
\label{sec:one-time-costs}

The one-time costs and continuous costs account for the investment required to
establish and execute the defect prediction model within a development process.
Both the one-time costs and the continuous costs depend on the current
development process in an organization, the desired defect prediction model, and
how the defect prediction is integrated into the development process.
Relevant factors of the current development process are, e.g., the availability
of links from work done in a version control system for the source to the issue
tracking system in use, as well as established procedures for the static
analysis of the software product. The factors that influence the resulting costs
of a defect prediction model are similar to those for any change in the
development process, e.g., tool costs, training costs, or migration costs.
Moreover, there is a trade-off between one-time costs and continuous costs:
high one-time costs can reduce the continuous costs. E.g., if a lot of money
is spent to create a defect prediction tool that runs fully automated, including
aspects like retraining and performance reports, the continuous costs are
relatively low in comparison to manual retraining and performance reporting.
Regardless, the estimation of the continuous costs should also account for
risks, e.g., additional training costs due to developer turnover.
Underestimating the risks may result in a too conservative estimate of the
continuous costs and, therefore, may require the re-estimation of the continuous
costs and a subsequent re-evaluation of the cost effectiveness of the defect
prediction model.

This already shows that the one-time and continuous costs are highly dependent
on the tool market: if off-the-shelf tools for defect prediction are available, the
costs are constant and depend on the licensing model, training costs of the
tool provider, as well es potential prior experience with tools by developers.

Moreover, the one-time costs and continuous costs are considered independently
from the quality assurance costs and costs of defects. We assume that the
one-time and continuous costs are constants that represent the mean of the
expected costs for establishing and executing the defect prediction model in an organization, i.e.,

\begin{equation}
\label{eq:cost-init}
cost_{INIT} = C_{INIT}
\end{equation}
and
\begin{equation}
\label{eq:cost-exec}
cost_{EXEC} = C_{EXEC}.
\end{equation}
 
% If an off-the-shelf tool is available, the costs are mostly fixed
% depending on the licensing model. Thus, these costs depend more on the existence
% of a tool market for defect prediction, than on the defect prediction models. In
% general, these costs are independent of other costs, i.e., the quality assurance
% costs $cost_{QA}$ and the costs of defects $cost_{DEF}$. They may only have an
% indirect relation, e.g., in case a ``cheaper'' defect prediction model is used
% that has worse performance. However, without a tool market, we cannot reasonably
% answer what the actual costs of integrating and maintaining these tools are.
% 
% Within this work, our focus is on developing a cost model, that can be used to
% compare the \ac{ROI} of defect prediction models. Since there is no tool market
% yet, we assume that the one-time costs and continuous costs for all defect
% prediction models are the same constant value 

\subsection{Quality assurance costs}
\label{sec:qa-costs}

These are the costs that result from the prediction of the defect prediction
model, i.e., that result from acting upon defective predictions by the model. 
The costs for the quality assurance measures depend on the techniques for
quality assurance (e.g., code reviews, software tests). Moreover, the quality assurance
costs may depend on the artifacts themselves and may vary, e.g., due to the
artifact size or complexity. The estimate of these costs should take the experience
of developers into account and, therefore, may need to be adopted in case of
developer turnover. We denote these costs as $qa(s)$. These costs occur
for all artifacts that are predicted as defective, i.e.,

\begin{equation}
\begin{split}
\label{eq:cost-qa1}
cost_{QA} &= \sum_{s \in S: h(s) = 1} qa(s).
\end{split}
\end{equation}

We model these costs in relation to the artifacts $S$ and not the
defects $D$, because defect prediction models also label artifacts, instead of
identifying concrete defects which may be related to a set of artifacts.

\subsection{Defect costs}
\label{sec:defect-costs}

The last component are the costs due to post-release defects $d$,
i.e., defects that are not found by the quality assurance and make it into the
wild. We denote the costs of each defect as $loss(d)$. The costs due to
post-release defects consist of two parts. The first part are the costs due to
defects, which are not found because they are not predicted by the model, i.e.,
$d \in D_{MISS}$. The second part of the costs are due to the imperfection of
quality assurance. Even if we predict an artifact as defective
and apply quality assurance measures, there is no guarantee that a defect is
actually found. We model this chance that quality assurance fails as $qf(d)$
where $qf(d) \in [0,1)$ is the expected value that the defect is not found by
the quality assurance. We assume that $qf(d) < 1$, i.e., that the quality
assurance has a chance to uncover all defects. The expected costs due to
defects that are missed by the quality assurance are $qf(d) \cdot loss(d)$ for
all artifacts $d \in D_{PRED}$. The expected costs due to post release defects
are

\begin{equation}
\begin{split}
\label{eq:cost-def1}
cost_{DEF} &= \sum_{d \in D_{MISS}} loss(d)
\\&+ \sum_{d \in D_{PRED}} qf(d) \cdot loss(d).
\end{split}
\end{equation}

\subsection{Complete general cost model}

If we use our definitions from equations
\eqref{eq:cost-init}-\eqref{eq:cost-def1} within equation \eqref{eq:cost1}, we
get

\begin{equation}
\label{eq:cost2}
\begin{split}
cost &= C_{INIT} + C_{EXEC} 
+ \sum_{s \in S: h(s) = 1} qa(s)
\\&+ \sum_{d \in D_{MISS}} loss(d)
+ \sum_{d \in D_{PRED}} qf(d) \cdot loss(d).
\end{split}
\end{equation}
 

%% file: conditions.tex
\section{Conditions for Cost-Saving Defect Prediction}
\label{sec:conditions}

A still unanswered question in defect prediction research is what it means for a
defect prediction model to be good and when defect prediction is actually
successful. To our mind a defect prediction is successful, if it saves costs, i.e., has a
positive profit. Whether a concrete defect prediction model is cost saving,
depends not only on the quality of the predictions, but also on the actual costs
of defects and the costs for quality assurance. These costs depend on the
project context. Therefore, general statements whether defect prediction models
are cost saving or not are impossible. However, if we assume that the costs for defects are a
known constant $C$, we can proof boundary conditions on $C$ that must be
fulfilled in order for the defect prediction model to have a positive profit.
That constant costs $C$ for defects are not unreasonable for practical purposes
is discussed in Section~\ref{sec:cost-defects}. Within this section, we derive boundary conditions on the costs
of defects $C$. 

With Theorem~\ref{thm:boundaries-random} we specify boundary
conditions that must be fulfilled to have a positive profit in comparison to
randomly applying quality assurance to artifacts with probability $p_{qa}$.
Theorem~\ref{thm:boundaries-random} specifies the boundary
conditions using the fraction of the sum of the costs saved due to predicted
defects ($\sum_{d \in D_{PRED}} (qf(d)-1)$) and the costs saved if defects are
randomly predicted correctly ($\sum_{d \in D} p_{qa}^{|d|} \cdot (1-qf(d))$) as
denominator and the costs for quality assurance for randomly predicted artifacts
($\sum_{s \in S} p_{qa} \cdot qa(s)$) minus the costs of quality assurance
due to a defect prediction model ($\sum_{s \in S: h(s) = 1} qa(s) - C_{INIT} -
C_{EXEC}$) as nominator. Depending on whether this ratio is positive or
negative, Theorem~\ref{thm:boundaries-random} defines an upper, respectively
lower boundary on the costs of defects.

\begin{theorem}\label{thm:boundaries-random}
Let $S$ be the artifacts of a software product with post-release defects $D$,
$h: S \to \{0,1\}$ a defect prediction model, $qa(s)$ the costs for the quality
assurance of artifact $s \in S$, $qf(d)$ the expected value that quality
assurances misses a defect $d$, $loss(d) = C$ the costs of a defect.
Furthermore, let $p_{qa}$ be the probability that quality assurance is applied
to an artifact randomly.

Let 
\begin{equation}
\begin{split}
\label{eq:x}
x = \sum_{d \in D_{PRED}} (qf(d)-1)
+\sum_{d \in D} p_{qa}^{|d|} \cdot (1-qf(d))
\end{split}
\end{equation}

\begin{equation}
\label{eq:y}
y = \sum_{s \in S} p_{qa} \cdot qa(s) - \sum_{s \in S: h(s) = 1} qa(s) -
C_{INIT} - C_{EXEC}
\end{equation}

The defect prediction model $h$ has an expected positive profit in comparison
to the random selection of artifacts with probability $p_{qa}$, if

\begin{equation}
\label{eq:boundary-random-x>0}
C < \frac{y}{x}~\text{if}~x>0 
\end{equation}
respectively 
\begin{equation}
\label{eq:boundary-random-x<0}
C > \frac{y}{x}~\text{if}~x<0
\end{equation}
\end{theorem}
\begin{proof}
To proof the boundaries on C, we analyze the expected profit of a defect
prediction model in comparison to randomly selecting artifacts, which is
\begin{equation}
profit = cost_{random} - cost
\end{equation}
where is $cost$ as defined in Equation~\eqref{eq:cost2} and $cost_{random}$ are
the costs of defect prediction if quality assurance is applied randomly with
probability $p_{qa}$. We get the $cost_{random}$ by applying the preconditions
of this theorem to Equation~\eqref{eq:cost2} and get
\begin{equation}
\label{eq:cost-random}
\begin{split}
cost_{random} &= \sum_{s \in S}p_{qa} \cdot qa(s)
+ \sum_{d \in D} (1-p_{qa}^{|d|}) \cdot loss(d) 
\\&+ \sum_{d \in D} p_{qa}^{|d|} \cdot qf(d) \cdot loss(d).
\end{split}
\end{equation}
Thus, we have
\begin{equation}
\begin{split}
profit &= \sum_{s \in S}p_{qa} \cdot qa(s)
+ \sum_{d \in D} (1-p_{qa}^{|d|}) \cdot loss(d) 
\\&+ \sum_{d \in D} p_{qa}^{|d|} \cdot qf(d) \cdot loss(d)
\\&- C_{INIT} - C_{EXEC}
- \sum_{s \in S: h(s) = 1} qa(s)
\\&- \sum_{d \in D_{MISS}} loss(d)
- \sum_{d \in D_{PRED}} qf(d) \cdot loss(d).
\end{split}
\end{equation}
Therefore, the profit is positive, if
\begin{equation}
\begin{split}
0 &< \sum_{s \in S}p_{qa} \cdot qa(s)
+ \sum_{d \in D} (1-p_{qa}^{|d|}) \cdot loss(d) 
\\&+ \sum_{d \in D} p_{qa}^{|d|} \cdot qf(d) \cdot loss(d)
\\&- C_{INIT} - C_{EXEC}
- \sum_{s \in S: h(s) = 1} qa(s)
\\&- \sum_{d \in D_{MISS}} loss(d)
- \sum_{d \in D_{PRED}} qf(d) \cdot loss(d).
\end{split}
\end{equation}
We now move all terms that contain $loss(d)$ to the left-hand side of the
equation and get

\begin{equation}
\begin{split}
\sum_{d \in D_{MISS}} loss(d)
&+\sum_{d \in D_{PRED}} qf(d) \cdot loss(d)
\\- \sum_{d \in D} (1-p_{qa}^{|d|}) \cdot loss(d)
&- \sum_{d \in D} p_{qa}^{|d|} \cdot qf(d) \cdot loss(d)
\\&<
\\\sum_{s \in S} p_{qa} \cdot qa(s)
- C_{INIT} &- C_{EXEC}
- \sum_{s \in S: h(s) = 1} qa(s)
\end{split}
\end{equation}
The right-hand side of the equation is now equal to $y$ as defined in
Equation~\eqref{eq:y}. For the left-hand side, we use that one of the conditions
of our theorem is $loss(d)=C$ and get
\begin{equation}
\begin{split}
\sum_{d \in D_{MISS}} C
&+\sum_{d \in D_{PRED}} qf(d) \cdot C
\\- \sum_{d \in D} (1-p_{qa}^{|d|}) \cdot C
&- \sum_{d \in D} p_{qa}^{|d|} \cdot qf(d) \cdot C
\\&< y
\end{split}
\end{equation}
Because $C$ is independent of the terms of the sums, we can factorize $C$ and
get
\begin{equation}
\begin{split}
C \cdot \Biggl(&\sum_{d \in D_{MISS}} 1+\sum_{d \in D_{PRED}} qf(d)
\\&-\sum_{d \in D} (1-p_{qa}^{|d|})
-\sum_{d \in D} p_{qa}^{|d|} \cdot qf(d)\Biggr) < y.
\end{split}
\end{equation}
Because $-\sum_{d \in D} (1-p_{qa}^{|d|}) = -\sum_{d \in D} 1 +
\sum_{d \in D} p_{qa}^{|d|})$ and $-\sum_{d \in D} 1 = -\sum_{d \in D_{miss}} 1
- \sum_{d \in D_{PRED}} 1$ it follows that
\begin{equation}
\begin{split}
C \cdot \Biggl(
&\sum_{d \in D_{PRED}} qf(d)
-\sum_{d \in D_{PRED}} 1
\\&+\sum_{d \in D} p_{qa}^{|d|}
-\sum_{d \in D} p_{qa}^{|d|} \cdot qf(d)
\Biggr) < y
\end{split}
\end{equation}
We can reformulate this as
\begin{equation}
\begin{split}
C \cdot \Biggl(
\sum_{d \in D_{PRED}} (qf(d)-1)
+\sum_{d \in D} p_{qa}^{|d|} \cdot (1-qf(d))
\Biggr) < y.
\end{split}
\end{equation}
The left-hand side of the equation is $C \cdot x$ with $x$ as defined in
Equation~\eqref{eq:x}. If we divide by $x$, we get our boundary conditions for $C$. 
\end{proof}

We use Theorem~\ref{thm:boundaries-random} to derive two corollaries.
Corollary~\ref{thm:boundaries-lowerboundary} specifies a lower boundary
through not applying additional quality assurance at all, which
is the same as a random defect prediction model with probablity $p_{qa}=0$. The
resulting boundary condition basically means that the costs saved due to
predicted defects must be greater than the costs for the additional quality
assurance. Corollary~\ref{thm:boundaries-upperboundary} specifies an
upper boundary through applying additional quality assurance to all artifacts,
which is the same as a random defect prediction model with probability $p_{qa}=1$.
The resulting boundary condition basically means that the costs due to missed
defects must be lower than the costs for the quality assurance for the artifacts
that were not predicted as defective would have been.

\begin{corollary}\label{thm:boundaries-lowerboundary}
Let $S$ be the artifacts of a software product with post-release defects $D$,
$h: S \to \{0,1\}$ a defect prediction model, $qa(s)$ the costs for the quality
assurance of artifact $s \in S$, $qf(d)$ the expected value that quality
assurances misses a defect, $loss(d) = C$ the costs of a defect. 

The defect prediction model $h$ has a positive profit in comparison to no
quality assurance for any artifact $s \in S$ if

\begin{equation}
\label{eq:C-boundary-lower}
C > \frac{\sum_{s \in S: h(s) = 1} qa(s) + C_{INIT} + C_{EXEC}}{\sum_{d \in
D_{PRED}} (1-qf(d))}. 
\end{equation}
\end{corollary}
\begin{proof}
No quality assurance is the same as a random approach for defect prediction with
$p_{qa} = 0$. We use this to calculate $x$ and $y$ from
Theorem~\ref{thm:boundaries-random}. 

For $y$ we get that 
\begin{equation}
y = \sum_{s \in S} 0 \cdot qa(s) - \sum_{s \in S: h(s) = 1} qa(s) -
C_{INIT} - C_{EXEC}
\end{equation}
Because the first term is zero, we get
\begin{equation}
y = -\sum_{s \in S: h(s) = 1} qa(s) - C_{INIT} - C_{EXEC}.
\end{equation}

For $x$ we get that
\begin{equation}
\begin{split}
x = \sum_{d \in D_{PRED}} (qf(d)-1)
+\sum_{d \in D} 0^{|d|} \cdot (1-qf(d))
\end{split}
\end{equation}
When we remove the terms that are now zero, we get
\begin{equation}
x = \sum_{d \in D_{PRED}} (qf(d)-1)
\end{equation}
Since $qf(d) \in [0,1)$ it follows that $qf(d)-1<0$ and consequently that $x<0$.
Thus, Equation~\eqref{eq:boundary-random-x<0} applies and we get
\begin{equation}
C > \frac{-\sum_{s \in S: h(s) = 1} qa(s) - C_{INIT} - C_{EXEC}}{\sum_{d \in
D_{PRED}} (qf(d)-1)}. 
\end{equation}
We can factorize -1 from the nominator of the right-hand side and then multiply
the -1 with the denominator instead and get
\begin{equation}
\label{eq:C-lower-proof-final}
C > \frac{\sum_{s \in S: h(s) = 1} qa(s) + C_{INIT} + C_{EXEC}}{\sum_{d \in
D_{PRED}} (1-qf(d))}. 
\end{equation} 
\end{proof}

\begin{corollary}\label{thm:boundaries-upperboundary}
Let $S$ be the artifacts of a software product with post-release defects $D$,
$h: S \to \{0,1\}$ a defect prediction model, $qa(s)$ the costs for the quality
assurance of artifact $s \in S$, $qf(d)$ the expected value that quality
assurances misses a defect, $loss(d) = C$ the costs of a defect. 

The defect prediction model $h$ has a positive profit in comparison to 
quality assurance for all artifacts $s \in S$ if

\begin{equation}
\label{eq:C-boundary-lower}
C < \frac{\sum_{s \in S: h(s) = 0} qa(s) - C_{INIT} - C_{EXEC}}{\sum_{d \in
D_{MISS}} (1-qf(d))}. 
\end{equation}
\end{corollary}
\begin{proof}
Quality assurance for all artifacts $s \in S$ is the same as a random approach
for defect prediction with $p_{qa} = 1$. We use this to calculate $x$ and $y$ from
Theorem~\ref{thm:boundaries-random}. 

For $y$ we get that 
\begin{equation}
y = \sum_{s \in S} 1 \cdot qa(s) - \sum_{s \in S: h(s) = 1} qa(s) -
C_{INIT} - C_{EXEC}
\end{equation}
Because $\sum_{s \in S} qa(s) = \sum_{s \in S: h(s)=0} qa(s) + \sum_{s \in S:
h(s) = 1} qa(s)$, we get
\begin{equation}
y = \sum_{s \in S: h(s) = 0} qa(s) - C_{INIT} - C_{EXEC}.
\end{equation}

For $x$ we get that
\begin{equation}
\begin{split}
x = \sum_{d \in D_{PRED}} (qf(d)-1)
 +\sum_{d \in D} 1^{|d|} \cdot (1-qf(d)).
\end{split}
\end{equation}
Because $\sum_{d \in D} (1-qf(d)) = \sum_{d \in D_{MISS}} (1-qf(d)) + \sum_{d
\in D_{PRED}} (1-qf(d))$ we get
\begin{equation}
\begin{split}
x = \sum_{d \in D_{MISS}} 1-\sum_{d \in D_{MISS}} qf(d).
\end{split}
\end{equation}
We can rewrite this as
\begin{equation}
\begin{split}
x = \sum_{d \in D_{MISS}} (1- qf(d)).
\end{split}
\end{equation}
Since $qf(d) \in [0,1)$ it follows that $1-qf(d)>0$ and consequently $x>0$.
Thus, Equation~\eqref{eq:boundary-random-x>0} applies and we get
\begin{equation}
\label{eq:C-boundary-lower}
C < \frac{\sum_{s \in S: h(s) = 0} qa(s) - C_{INIT} - C_{EXEC}}{\sum_{d \in
D_{MISS}} (1-qf(d))}. 
\end{equation}
\end{proof}

%% file: init.tex
\section{Initializations of the general cost model}
\label{sec:init}

To actually initialize our cost model and use it for the computation of costs
and concrete values for the boundary conditions on defect costs, we need
estimations for $C_{INIT}$, $C_{EXEC}$, $qa(s)$, $qf(s)$, and $loss(d)$. We
start this section with a discussion on how each of these costs components may
be estimated. Then, we proceed to initialize six concrete cost models from the
general costs model including the resulting boundary conditions on the costs of
defects. 

\subsection{One-time and execution costs}

In Section~\ref{sec:one-time-costs}, we already argued that the costs $C_{INIT}$ and
$C_{EXEC}$ are rather an issue of the development of a tool market, than of the
actual performance of the defect prediction model. Moreover, tooling costs are
usually small in comparison to other human resources. Because these costs are
likely a minor costs component and not a deciding factor for the introduction of
a defect prediction model, we assume these costs to be zero and omit them from
the initializations of the general cost model. In case the costs are known for a
specific use case, the constants can simply be added to the model again, without
major changes. The impact of this decision is further discussed in
Section~\ref{sec:impact-exec-costs}.

\subsection{Cost of quality assurance}

For the size-aware cost model, we assume that the costs of quality assurance are
a linear function of the size of a software artifact, i.e.,

\begin{equation}
\label{eq:qa-size}
qa_{size}(s) = C_{QA} \cdot size(s)
\end{equation}
where $C_{QA}$ is a constant that describes the quality assurance costs per
size unit. The idea to measure quality assurance effort in relation to the
size of software artifacts is, e.g., also used by Rahman\etal~\cite{Rahman2012}
and Canfora\etal~\cite{Canfora2013, Canfora2015}, who both use the lines of
code of the software artifacts as indicator for quality assurance efforts. 

In case a simpler approach is wanted or the size of artifacts is not considered
a relevant factor, we propose to use constant costs for quality assurance
independent of the size, i.e.,
\begin{equation}
qa_{const}(s) = C_{QA}.
\end{equation}

The assumption of constant costs is also used, e.g., in the cost models by
Khoshgoftaar\etal~\cite{Khoshgoftaar1998, Khoshgoftaar2004} and Zhang and
Cheung~\cite{Zhang2013}. In both cases, the values for the constants in a
concrete use of the cost model must be estimated by a member of the
organization that wants to use the model based on the quality assurance measures
that are applied.

\subsection{Costs of defects}
\label{sec:cost-defects}

The costs of defects are hard to estimate. To the best of our knowledge, there
is no empirical evidence on the costs of post-release defects in general.
Anecdotal evidence suggests that the costs of post-release defects depend
strongly on the project context and the defects themselves. Some defects are
extremely costly, others may be very cheap. In conclusion, there is no
general way to estimate the costs of defects. To the best of our knowledge,
there is also no study that links the costs of defects to
specific features, e.g., the source code or similar.
Because of that, the (currently) only reasonable way for estimating the costs of
defects is to use a constant for the costs of defects, i.e., 

\begin{equation}
\label{eq:loss-const}
loss(d) = C_{DEF}.
\end{equation}

This constant reflects the mean costs of defects within a project. Same as for
$C_{QA}$, the value of this constant should be estimated by a member of the
organization based on data about past defects. 

\subsection{Quality assurance failures}

We use Bernoulli experiments to determine if quality assurance fails to reveal a
defect $d$ in artifact $s$, i.e., we assume that we have a
probability of $p_{qf}$ that the quality assurance does not discover a defect
$d$ in $s$, independent of $s$ itself. In order to successfully prevent a post-release
defect $d$, the quality assurance must not fail to reveal $d$ on all artifacts
$s \in d$. Thus, we have to perform $|d|$ Bernoulli experiments and the quality
assurance fails if any of the Bernoulli experiments fails. This is the opposite
of $|d|$ times not failing the Bernoulli experiment, which has a probability of
$(1-p_{qf})^{|d|}$. Thus, the probability of not finding a defect is
$1-(1-p_{qf})^{|d|}$. It follows that we can also do a single Bernoulli experiment
with probability $1-(1-p_{qf})^{|d|}$. Since the
expected value of the success of repeated Bernoulli experiments is the same as
the probability of the Bernoulli experiment, we get
\begin{equation}
qf(d) = 1-(1-p_{qf})^{|d|}.
\end{equation}

\subsection{Cost ratios}
\label{sec:ratios}

For the above definitions of costs, we use the constants $C_{QA}$ and $C_{DEF}$
to model the average costs of quality assurance, respectively defects. Let $C$
be the ratio between the average costs, i.e., $C = \frac{C_{DEF}}{C_{QA}}$. This
estimation of costs as a ratio is based on the work by
Khoshgoftaar\etal~\cite{Khoshgoftaar1998, Khoshgoftaar2004}. Since we do
not know the actual project-specific costs, we can assume without loss of
generality that $C_{QA} = 1$, which means that we assume that our quality
assurance constant is ``one cost unit''. We then get $C = \frac{C_{DEF}}{1} =
C_{DEF}$. Thus, this ratio is consistent with the costs for defects for which we
defined boundary conditions in Theorem~\ref{thm:boundaries-random},
Corolarry~\ref{thm:boundaries-lowerboundary}, and
Corolarry~\ref{thm:boundaries-upperboundary}. We can use $C$ to reformulate
equations~\eqref{eq:qa-size}--\eqref{eq:loss-const} as
\begin{equation}
\label{eq:ratios}
\begin{split}
&qa_{size}(s) = C_{QA} \cdot size(s) = size(s)\\
&qa_{const}(s) = C_{QA} = 1 \\
&loss(d) = C_{DEF} = C.
\end{split}
\end{equation}

Through this ratio, organizations also do not need to estimate $C_{QA}$ and
$C_{DEF}$ directly anymore. An estimate how much the costs of defects is in
relation to the costs for the quality assurance is sufficient.
 
\subsection{Relationship between artifacts and defects}

Another factor that influences how we initialize the general cost model is the
actual relationship between artifacts and defects. There are three relevant
scenarios.
\begin{enumerate}
  \item The $n$-to-$m$ is the scenario we considered so far, i.e., each
  defect may belong to multiple artifacts and each artifact may contain
  multiple defects. This is the most complex, but also most realistic scenario.
  The importance of this scenario was also highlighted by
  Hemmati\etal~\cite{Hemmati2015} in their work on exploiting the $n$-to-$m$
  relationships for improving rankings of defect prediction models. 
  \item In the $1$-to-$m$ scenario, each defect may only belong to one
  artifact, but artifacts may have multiple defects. This scenario is relevant, because it
  reflects defect prediction data sets with bug counts, e.g,
  \cite{DAmbros2012}.
  \item In the $1$-to-$1$ scenario, each artifact is either labeled as
  defective or not. All information regarding the number of artifacts that are
  impacted by a defect or the number of defects per artifact is ignored. Some
  data sets contain this kind of data, e.g., the NASA MDP
  data.\footnote{http://openscience.us/repo/defect/mccabehalsted/}
  Moreover, this scenario is dominant in the evaluation of defect prediction approaches in
  the literature, as can, e.g., be seen in the analysis which
  metrics were used for the evaluation of studies on cross-project
  defect prediction~\cite{Hosseini2017,Herbold2017}.
\end{enumerate}
 
Theoretically, we could also consider the $n$-to-$1$ scenario, i.e., each defect
may belong to multiple artifacts, but each artifact may only be affected by one
defect. To the best of our knowledge, this scenario was never considered in
defect prediction research so far. Moreover, defect data sets with bug counts
demonstrate that there are files that are affected by multiple defects, i.e.,
that this scenario is unrealistic. Therefore, we do not consider this
relationship any further.

\subsection{Initializations of the cost model}

Due to the two ways to model the quality assurance costs (size-aware and
constant), and the three relationships between artifacts and defects
($n$-to-$m$, $1$-to-$m$, $1$-to-$1$), we get a total of six initializations of
our general cost model.

We start with the initialization of the general cost model from
Equation~\eqref{eq:cost2} using the size-aware quality assurance costs
$qa_{size}$ and the loss function from Equation~\eqref{eq:ratios}.
This way, we get a size-aware cost model with an $n$-to-$m$ mapping between
artifacts and defects  
\begin{equation}
\label{eq:cost-size-n-to-m-eq1}
\begin{split}
cost_{size,n/m} &= \sum_{s \in S: h(s) = 1} size(s)
\\&+ \sum_{d \in D_{MISS}} C
+ \sum_{d \in D_{PRED}} (1-(1-p_{qf})^{|d|}) \cdot C.
\end{split}
\end{equation}

Because $\sum_{d \in D_{MISS}} C = |D_{MISS}| \cdot C$, we can rewrite
Equation~\eqref{eq:cost-size-n-to-m-eq1} equation as
\begin{equation}
\label{eq:cost-size-n-to-m-eq2}
\begin{split}
cost_{size,n/m} &= \sum_{s \in S: h(s) = 1}  size(s)
\\&+ |D_{MISS}| \cdot C
+ \sum_{d \in D_{PRED}} (1-(1-p_{qf})^{|d|}) \cdot C.
\end{split}
\end{equation}

We get a 1-to-$m$ relationship between software artifacts and defects if we
assume $|d|=1$ for all $d \in D$. We observe that 
\begin{equation}
(1-(1-p_{qf})^{1}) = p_{qf}.
\end{equation}
When we use this to simplify Equation~\eqref{eq:cost-size-n-to-m-eq2}, we get a
size-aware cost model with a $1$-to-$m$ mapping between artifacts and defects
\begin{equation}
\begin{split}
cost_{size,1/m} &= \sum_{s \in S: h(s) = 1}  size(s)
\\&+ |D_{MISS}| \cdot C
+ \sum_{d \in D_{PRED}} p_{qf} \cdot C.
\end{split}
\end{equation}
Because $\sum_{d \in D_{PRED}} p_{qf} \cdot C = |D_{PRED}| \cdot p_{qf} \cdot
C$, we can rewrite this as
\begin{equation}
\label{eq:cost-size-1-to-m}
\begin{split}
cost_{size,1/m} &= \sum_{s \in S: h(s) = 1}  size(s)
\\&+ |D_{MISS}| \cdot C
+ |D_{PRED}| \cdot p_{qf} \cdot C.
\end{split}
\end{equation}

\input{datastats_modified.tex}

If we further assume that there is only one defect per artifact, it follows that
$|D_{MISS}|=fn$ and $|D_{PRED}|=tp$. Using this, we can further simplify the
cost model from  Equation~\eqref{eq:cost-size-1-to-m} and get a size-aware cost model
with a $1$-to-$1$ mapping between artifacts and defects
\begin{equation}
\label{eq:cost-size-1-to-1}
\begin{split}
cost_{size,1/1} &= \sum_{s \in S: h(s) = 1}  size(s)
\\&+ fn \cdot C
+ tp \cdot p_{qf} \cdot C.
\end{split}
\end{equation}

If we initialize the general cost model from Equation~\eqref{eq:cost2} with
$qa_{const}$ from Equation~\ref{eq:ratios}, we get a constant cost model where the artifact
size is not taken into account and a $n$-to-$m$ relationship between artifacts
and defects
\begin{equation}
\label{eq:cost-const-n-to-m-eq1}
\begin{split}
cost_{const,n/m} &= \sum_{s \in S: h(s) = 1} 1
\\&+ \sum_{d \in D_{MISS}} C
+ \sum_{d \in D_{PRED}} (1-(1-p_{qf})^{|d|}) \cdot C.
\end{split}
\end{equation}

Because $\sum_{s \in S: h(s) = 1} 1 = tp+fp$ and $\sum_{d \in D_{MISS}} C =
|D_{MISS}| \cdot C$ we can rewrite Equation~\eqref{eq:cost-size-n-to-m-eq1} as
\begin{equation}
\label{eq:cost-const-n-to-m-eq2}
\begin{split}
cost_{const,n/m} &= tp+fp
\\&+ |D_{MISS}| \cdot C
\\&+ \sum_{d \in D_{PRED}} (1-(1-p_{qf})^{|d|}) \cdot C.
\end{split}
\end{equation}

We can simplify Equation~\eqref{eq:cost-const-n-to-m-eq2} analogously to
Equation~\eqref{eq:cost-size-1-to-m} and get a constant cost model with a
$1$-to-$m$ mapping between artifacts and defects
\begin{equation}
\label{eq:cost-const-1-to-m}
\begin{split}
cost_{const,1/m} &= tp+fp
\\&+ |D_{MISS}| \cdot C
+ |D_{PRED}| \cdot p_{qf} \cdot C.
\end{split}
\end{equation}

Similarly, we can simplify Equation~\eqref{eq:cost-const-1-to-m} analogously to
Equation~\eqref{eq:cost-size-1-to-1} and get a constant cost model with a
$1$-to-$1$ relationship between artifacts and defects.
\begin{equation}
\label{eq:cost-const-1-to-1}
\begin{split}
cost_{const,1/1} &= tp+fp
\\&+ fn \cdot C
+ tp \cdot p_{qf} \cdot C.
\end{split}
\end{equation}

The above initializations of the cost model can be used to calculate the costs
only if an estimate for the ratio between the costs for defects and the costs
for quality assurance $C$ is available. If an organization cannot estimate
these costs, we can apply these initializations to
corollaries~\ref{thm:boundaries-lowerboundary} and
\ref{thm:boundaries-upperboundary} and derive cost boundaries that define for
which cost ratios $C$ a defect prediction model would be cost saving.
Organizations could then estimate if it is likely that their ratio $C$ is within the
boundaries and determine if the defect prediction model has cost saving
potential, even though they cannot determine the amount of the cost savings directly. The
following corollary formalizes this and establishes the boundary conditions for
the different initializations of the cost model.

\begin{corollary}
\label{thm:boundaries-init}
Given the cost functions $cost_{size, n/m}$, $cost_{size, 1/m}$, $cost_{size,
1/1}$, $cost_{const,n/m}$, $cost_{const,1/m}$, respectively $cost_{const,1/1}$,
a defect prediction model $h: S \to \{0,1\}$ has a positive profit given
software artifacts $S$ and defects $D$ if 
\begin{equation}
\begin{split}
\frac{\sum_{s \in S: h(s)=1} size(s)}{\sum_{d \in D_{PRED}} (1-p_{qf})^{|d|}}
&< C_{size, n/m} <
\frac{\sum_{s \in S: h(s)=0} size(s)}{\sum_{d \in D_{MISS}} (1-p_{qf})^{|d|}}
\\
\frac{\sum_{s \in S: h(s)=1} size(s)}{|D_{PRED}| (1-p_{qf})}
&< C_{size, 1/m} <
\frac{\sum_{s \in S: h(s)=0} size(s)}{|D_{MISS}| (1-p_{qf})}
\\
\frac{\sum_{s \in S: h(s)=1} size(s)}{tp \cdot (1-p_{qf})}
&< C_{size, 1/1} <
\frac{\sum_{s \in S: h(s)=0} size(s)}{fn \cdot (1-p_{qf})}
\\
\frac{tp+fp}{\sum_{d \in D_{PRED}} (1-p_{qf})^{|d|}}
&< C_{const, n/m} <
\frac{tn+fn}{\sum_{d \in D_{MISS}} (1-p_{qf})^{|d|}}
\\
\frac{tp+fp}{|D_{PRED}| (1-p_{qf})}
&< C_{const, 1/m} <
\frac{tn+fn}{|D_{MISS}| (1-p_{qf})}
\\
\frac{tp+fp}{tp \cdot (1-p_{qf})}
 &< C_{const,1/1} <
\frac{tn+fn}{fn \cdot (1-p_{qf})},
\end{split}
\end{equation}
with $C_{size, n/m}, C_{size, 1/m}, C_{size, 1/1}, C_{const, n/m}, C_{const,
1/m}$, and $C_{const, 1/1}$ the ratios between the costs of defects
and costs for quality assurance, respectively.
\end{corollary}
\begin{proof}
The boundaries follow directly from the
Corollary~\ref{thm:boundaries-lowerboundary},
Corollary~\ref{thm:boundaries-upperboundary} and the definitions and
calculations from Section~\ref{sec:init}.
\end{proof}

%% file: datastats_modified.tex
% latex table* generated in R 3.5.1 by xtable* 1.8-3 package
% Wed Sep 12 21:48:31 2018
\begin{table}[t]
\centering
\begin{tabular}{p{1.4cm}rrrrr}
  \hline
 & $|S|$ & $|S_{DEF}|$ & $|D|$ & $mean(|d|)$ & $mean(LOC)$ \\
  \hline
%  archiva & 746 &  11 &   4 & 3.25 & 123.48 \\ 
%  calcite & 1543 & 716 & 235 & 7.63 & 169.36 \\ 
%  cayenne & 3613 & 536 & 103 & 6.65 & 70.06 \\ 
%  commons- & \multirow{2}{*}{1309} & \multirow{2}{*}{347} & 
% \multirow{2}{*}{24} & \multirow{2}{*}{23.62} & \multirow{2}{*}{136.76} \\
% ~~math & & & & & \\ 
%  deltaspike & 1730 &  41 &  24 & 2.17 & 46.44 \\ 
%  falcon & 850 &  64 &  52 & 3.46 & 143.89 \\ 
%  kafka & 1584 & 638 & 313 & 4.83 & 114.67 \\ 
%  kylin & 1384 & 279 & 230 & 2.33 & 101.88 \\ 
%  nutch & 520 & 269 &  69 & 9.22 & 102.33 \\  
%  storm & 2277 & 1998 & 333 & 26.36 & 113.82 \\ 
%  struts & 1950 & 1800 & 108 & 36.80 & 79.45 \\ 
%  tez & 1062 & 191 & 102 & 3.16 & 169.01 \\ 
%  tika & 981 & 157 &  83 & 3.63 & 107.25 \\ 
%  wss4j & 715 &  30 &  13 & 2.77 & 153.78 \\ 
%  zeppelin & 563 & 238 & 172 & 2.97 & 163.79 \\ 
%  zookeeper & 622 &  99 &  46 & 2.78 & 124.49 \\ 
%  \hline
  archiva & 508 &   6 &   4 & 2.00 & 108.85 \\ 
  cayenne & 2121 & 281 &  74 & 5.12 & 73.46 \\ 
  commons- & \multirow{2}{*}{789} &   \multirow{2}{*}{2} &   \multirow{2}{*}{2}
  & \multirow{2}{*}{1.00} & \multirow{2}{*}{112.94} \\
  math & & & & & \\ 
  deltaspike & 793 &  14 &  13 & 1.31 & 56.17 \\ 
  falcon & 577 &  38 &  33 & 2.91 & 121.82 \\ 
  kafka & 1119 & 201 & 212 & 2.00 & 87.54 \\ 
  kylin & 1094 & 170 & 138 & 1.95 & 105.98 \\ 
  nutch & 414 &  37 &  30 & 1.73 & 106.74 \\ 
  storm & 1981 & 173 & 138 & 1.88 & 114.68 \\ 
  struts & 1334 &  61 &  38 & 2.26 & 79.36 \\ 
  tez & 803 & 94 &  71 & 1.98 & 129.33 \\ 
  tika & 694 &  44 &  35 & 1.62 & 105.06 \\ 
  wss4j & 501 &  10 &  7 & 2.00 & 110.55 \\ 
  zeppelin & 394 & 89 & 142 & 1.63 & 177.53 \\ 
  zookeeper & 380 &  41 &  27 & 1.85 & 113.21 \\ 
   \hline
\end{tabular}
\caption{Data used for the evalulation of the cost model.}
\label{tbl:data}
\end{table}

%% file: experiments.tex
\section{Experiments}
\label{sec:experiments}

While the focus of this article is a theoretical model, we also want to get
insights into the practical relevance of the model. The main difference between
our approach and the state of the art is that consider an $n$-to-$m$
relationship between artifacts and defects. Thus, the primary goal of these
experiments is to evaluate if there are difference in real-world data between
the 1-to-1, the 1-to-$m$ and the $n$-to-$m$ cost models. Differences between the
1-to-1 and 1-to-$m$ cost model should manifest for all projects, in which there
are artifact that are affected by multiple defects. Differences between 1-to-$m$
and $n$-to-$m$ should manifest for all projects, in which there are defects that
affect multiple files. A secondary goal of these experiments is to get insights
into the empirical relationship between confusion matrix based metrics and the
cost boundaries. Through this, we want to answer the question if metrics like
\textit{precision} and \textit{recall} may be sufficient to evaluate the cost
effectiveness of defect prediction models.

\subsection{Data}
\label{sec:data}

We used SmartSHARK~\cite{Trautsch2017} to collect data from a convenience sample of fifteen Apache
projects to conduct our experiments.\footnote{Detailed steps of the data collection are part of the
replication kit~\cite{replication-kit}.}. For each project, we collected the
commits for the year 2017 from the master branch of the repository. Then, we
identified the links between commits and issues in the Jira issue tracker of the
project. We labelled commits as fixing defects, if referenced an issue of type
``bug'', that has the status ``resolved'' or ``closed'' or had this status at any point in their lifetime,
and that is not a ``duplicate'' of another issue. Once we identified which
defects were fixed, we used the hunks of the commits to identify which files
were changed during the correction of an issue. We filtered the hunks to exclude
changes that only affected whitespaces or comments. In case the same issue was
referenced as part of multiple commits, the hunks for all referencing commits
were used. We created a matrix for each project that has as rows the file names
and as columns the defects that were fixed, and the entries depict whether a
file was part of a bugfixing commit for the issue. We only allowed files ending
with .java as part of the data set and furthermore used a heuristic to exclude
tests\footnote{Files that were in a path that included a folder called
``test'' were excluded}. Defects that were fixed in the year 2017 but that did
not lead to any change in a Java file were ignored. We did not use any
keyword-based approach for the identification of bugfixing commits, as these
cannot distinguish between pre-release defects and post-release defects. By
using only labels based on issues we ensure that our data contains only
post-release defects. We manually validated for all issues that affected more
than one Java files that they were really defects and discarded all issues that
are mislabeled as defect in the issue tracking system. The manual validation
followed the criteria for defects established by Herzig\etal~\cite{Herzig2013}.
This way, we discarded 175 issues of 413 issues that affect multiple files. We
restricted our manual validation to issues affecting multiple files due to the
high manual effort required for the validation of issue types. By only excluding 
issues that affect multiple files, we bias the evaluation against showing a
difference between the $n$-to-$m$ model on the one hand, and the 1-to-1 and the
1-to-$m$ model on the other hand.

Table~\ref{tbl:data} gives an overview over the projects we analyzed. The table
shows the number of files of the project ($|S|$), the number of files that were
affected by any defect ($|S_{DEF}|$), the number of defects that were fixed
($|D|$), the mean number of files affected by each defect ($mean(|d|)$), and the
mean logical \ac{LOC}\footnote{Non-empty lines that are not only comments.}.

\subsection{Simulation of defect prediction}

The goal of our experiments is to get insights into our cost model, especially
the impact of the different relationships between defects and artifacts (1-to-1,
1-to-$m$, $n$-to-$m$) on the boundary conditions for cost effective defect
prediction. To achieve this, we simulated classification models for defect prediction
that achieve different performances on the data. We performed a Bernoulli
experiment for each software artifact with the expected accuracy as probability
to simulate the defect prediction. If the experiment is successful, we assign
the correct label, if it fails we flip the label. We used the values $0.05$ to
$0.95$ with a step size of $0.05$ for the expected accuracy. To account for the
randomness of the labelling, we repeated this 100 times. We use the formulas
from Corollary~\ref{thm:boundaries-init} to calculate the boundaries on the cost
efficiency for the defect prediction for each simulation run.

\subsection{Results}
\label{sec:results}

Because our results are based on simulated defect prediction models and we are
only interested in the general trends, we do not report the exact values of the
simulation within this article, but only perform a visual analysis. 
Due to space restrictions, we cannot include all results in this manuscript.
However, plots for all simulations, the raw simulation results, the defect data
we collected for the projects, as well as the source code for the simulation
and the generation of the plots can be found in our replication
kit~\cite{replication-kit}.

Figure~\ref{fig:rep-results}(a)--(b) show representative results of the
simulations with a perfect quality assurance, i.e., $p_{qf}=0$. The plots depict how the
upper and lower boundaries of the different cost models evolve with respect to the metrics
\textit{recall} ($\frac{tp}{tp+fn}$) and \textit{precision}
($\frac{tp}{tp+fp}$). We choose to show \textit{recall} and \textit{precision}
here, because these are the most commonly used metrics for defect prediction
research~\cite{Hosseini2017}. The different colors show the data for the
different relationships between software artifacts and defects. 
We identified two types of projects regarding the trends for the boundaries,
that are distinguished by the required values for \textit{recall} and
\textit{precision} for the model to be cost saving.
\begin{itemize}
  \item Projects where defect prediction can be cost saving with a
  high value for \textit{recall} and a very low \textit{precision} ($<25\%$).
  The projects archiva, cayenne, commons-math, deltaspike, falcon, kylin, nutch, storm, struts,
  tez, tika, wss4j, and zookeeper belong to this category for which the difficulty of cost
  saving defect prediction is low.
  \item Projects where defect prediction can be cost efficient with a high value
  for \textit{recall} and a mediocore \textit{precision} ($>25\%$ and $<50\%$).
  The projects kylin and zeppelin belong to this category for which the
  difficulty of cost saving defect prediction is medium.
\end{itemize}
We note that there are no projects that require a high precision for cost
effective defect prediction in our data. 

Figure~\ref{fig:rep-results}(c) shows the result for falcon with
$p_{qf}=0.5$, i.e., a fifty percent chance that a defect is missed in an
artifact regardless of the additional quality assurance. This result is
representative result for the effect of imperfect quality assurance ($p_{qf}>0$)
on the cost boundaries. For all projects, we observe two changes with increasing
values of $p_{qf}$. First, the cost ratios $C$ for which defect prediction may
be cost saving increases with $p_{qf}$. This is expected, because $(1-p_{qf})$
is part of the denominator of all cost boundaries. Second, cost saving defect
prediction models can be achieved with a lower performance of the defect
prediction model for the $n$-to-$m$ relationship. This is likely due the fact
that the denominator contains $(1-p_{qf})^{|d|}$ for the $n$-to-$m$. This gives
more weight to defects that affect only few artifacts, i.e., that are easier to
predict.

\begin{figure*}
\centering
\includegraphics[width=0.8\textwidth]{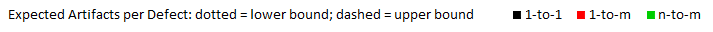}

\footnotesize
(a) falcon - very low difficulty defect prediction

\begin{minipage}{0.46\textwidth}
\centering
\includegraphics[width=\textwidth]{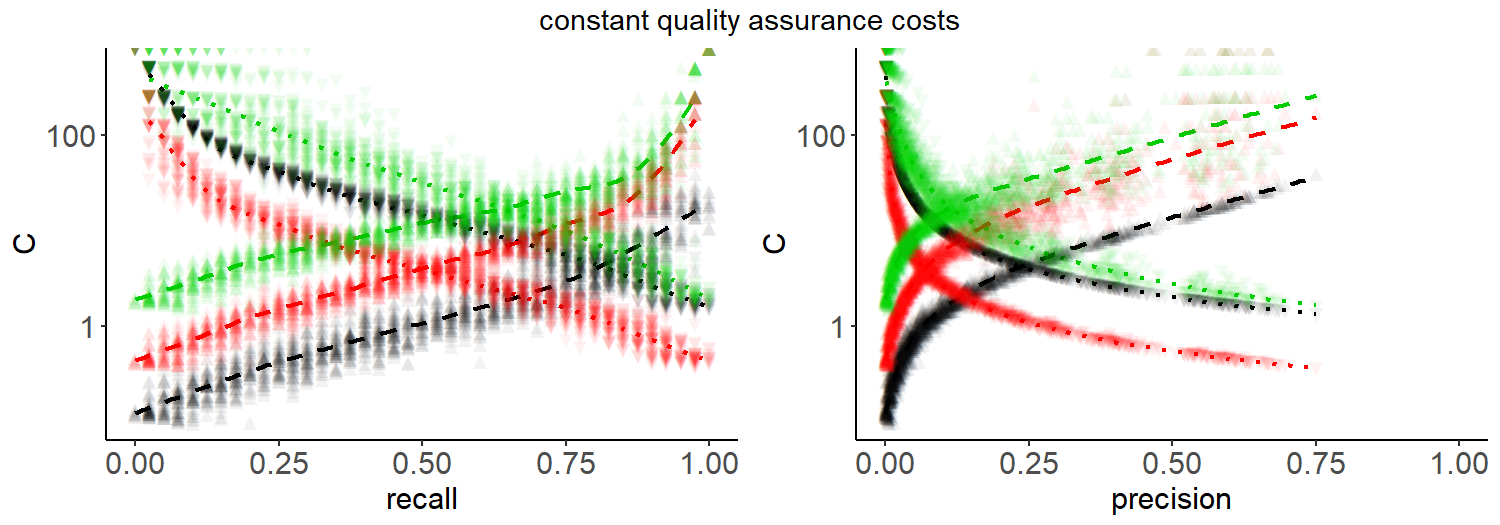}
\end{minipage}
\begin{minipage}{0.46\textwidth}
\centering
\includegraphics[width=\textwidth]{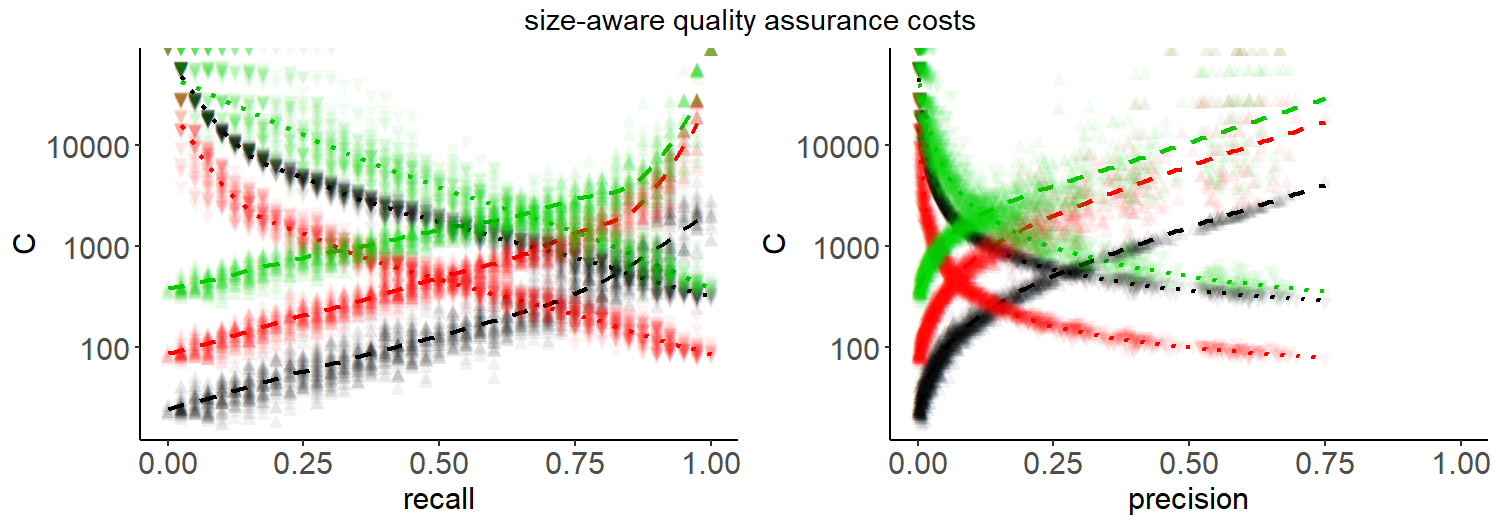}
\end{minipage}

(b) zeppelin - medium difficulty defect prediction

\begin{minipage}{0.46\textwidth}
\centering
\includegraphics[width=\textwidth]{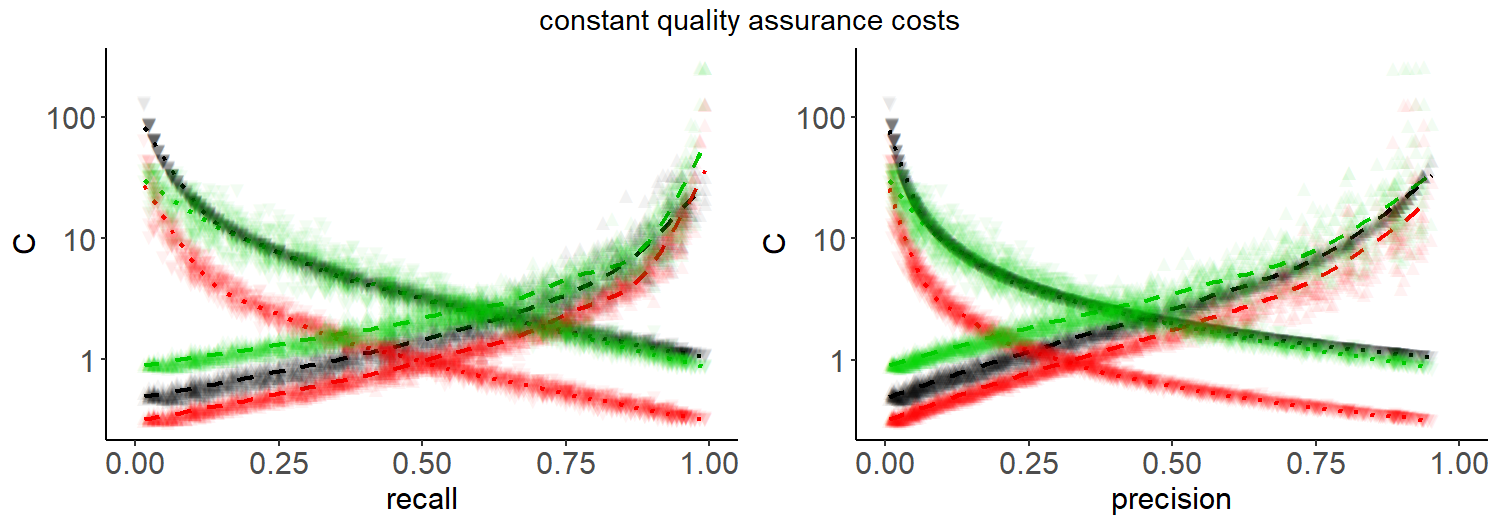}
\end{minipage}
\begin{minipage}{0.46\textwidth}
\centering
\includegraphics[width=\textwidth]{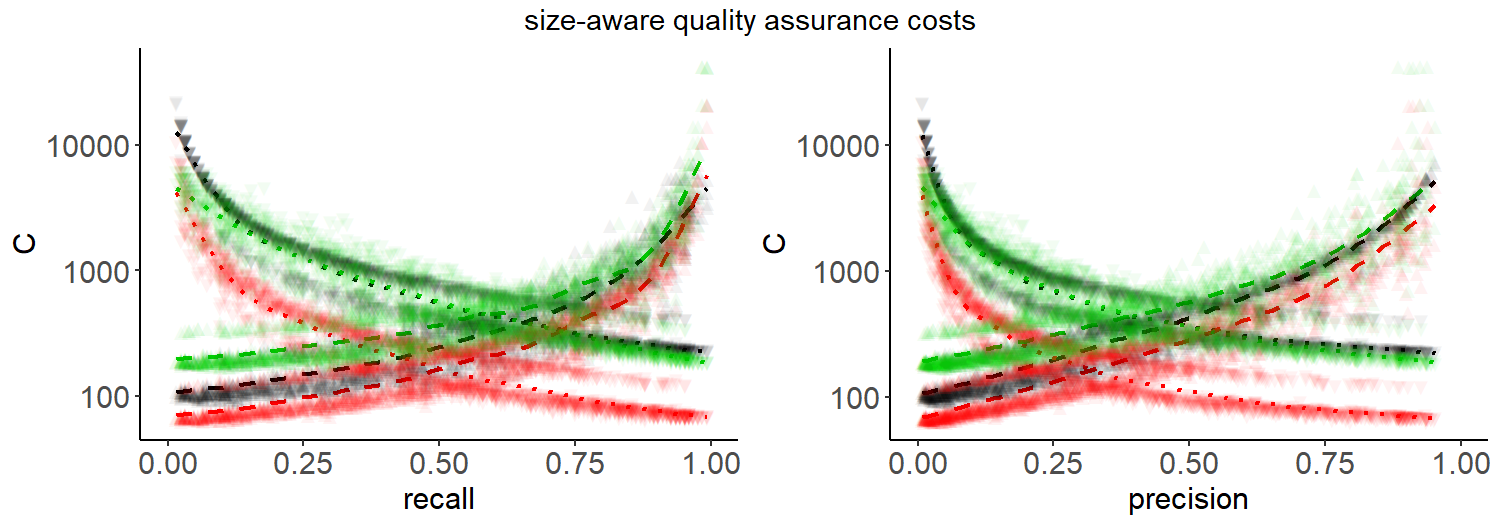}
\end{minipage}

(c) falcon - imperfect quality assurance with $p_{qf}=0.5$

\begin{minipage}{0.46\textwidth}
\centering
\includegraphics[width=\textwidth]{{{falcon_0.5_const}}}
\end{minipage}
\begin{minipage}{0.46\textwidth}
\centering
\includegraphics[width=\textwidth]{{{falcon_0.5_size}}}
\end{minipage}
\caption{Trends of the upper and lower boundaries on the cost
effectiveness. The lines show the regression of how the boundaries evolve in for
changing values of \textit{recall} and \textit{precision}. The jitter shows the
actual values of the experiments.}
\label{fig:rep-results}
\end{figure*}

%\caption{Trends of the boundaries on the cost effectiveness for $p_{qf}=0.5$.
%The lines show the regression of how the boundaries evolve in for changing
% values of \textit{recall} and \textit{precision}. The jitter shows the actual
% values of the experiments.}
%\label{fig:results-imperfect}
%\end{figure*}

%% file: discussion.tex
\section{Discussion}
\label{sec:discussion}

While we evaluated our cost model only on simulated data and not on real defect
prediction data sets, there are already several important insights that
highlight the need for such a cost model for accurate evaluations of defect
prediction models. 

\subsection{Impact of n-to-m relationships}
\label{sec:n-to-m-relationship}

There are large differences between the $n$-to-$m$, the
$1$-to-$m$ and the 1-to-1 cost model. This is true both for the required
performance of the models that allow for cost efficient predictions (i.e., lower boundary is
less than upper boundary) as well as the range of values $C$ for which the
prediction is cost saving. This effect is present for all projects we analyzed.

Thus, the results demonstrate that evaluations of defect prediction models may
lead to wrong conclusions if the $n$-to-$m$ relationship between defects and
artifacts is ignored. The estimated intervals for the cost ratio between defects
and quality assurance for which defect prediction is cost efficient differs
strongly from assuming 1-to-1 relationships. We expected this result due to
multiple reasons:
\begin{itemize}
  \item In the 1-to-1 relationship, the same defects may be counted as
  multiple distinct defects, i.e., once for each artifact that is affected. 
  \item In the 1-to-1 relationship, defects may be ignored, i.e., if one
  artifact is affected by multiple defects. 
\end{itemize}

To the best of our knowledge, there is no publication on defect prediction, that
evaluates the results with respect to the $n$-to-$m$ relationship. This problem
affects all current defect prediction research. To the best of our knowledge,
there is only the work by Hemmati\etal~\cite{Hemmati2015}, that exploits the
$n$-to-$m$ relationship for the generation of a ranking of files. However,
Hemmati\etal{} evaluate their study based on the density of remaining
defects per file, i.e., they use an effort-aware 1-to-$m$ metric
for the the evalution of their results. We note that all approaches from the state of the art can be evaluated with
respect to the $n$-to-$m$ relationship, even if they completely ignore this property.
This is possible, because our cost model only requires a binary function for
labeling of single artifacts as defined in
Equation~\eqref{eq:defect-prediction-model}. Regardless, algorithms that would
already consider the $n$-to-$m$ relationship while fitting a prediction model,
as is done by Hemmati\etal~\cite{Hemmati2015}, are likely to perform better.

Moreover, the currently publicly
available data sets do not contain the required information to evaluate
$n$-to-$m$ relationships instead of 1-to-1 relationships. While some popular
data sets (e.g., \cite{Jureczko2010, DAmbros2010}) contain defect counts and can
be used for evaluations respecting the 1-to-$m$ relationship, others contain
only binary labels and thus can only be used for 1-to-1 relationships (e.g., the
NASA MDP data). The data published together with the RELINK approach is the only
exception \cite{Wu2011}. The prepared data for defect prediction only contains
binary labels and, thus, only allows the evaluations with 1-to-1 relationships.
However, the provided meta-data contains the links between all issues and
commits, as well as the files that were touched in each commit. Thus, the
generation of $n$-to-$m$ relationships would be possible with further processing
of the meta data.

However, RELINK only contains data about four projects and only for three of
them the complete data required for defect prediction are available. This is too
few for realistic evaluations and, therefore, cannot be used to resolve this
problem for our community. To resolve this problem, we require new public data
sets for defect prediction research. With such new data, we can improve our
evaluations to take the $n$-to-$m$ relationship into account and investigate the
severity of ignoring this aspect in the last decades.

\subsection{Relationship to confusion matrix based metrics}
\label{sec:confusion-matrix-metrics}

Our cost model is not unrelated to other performance measures. Our simulations
already show that the boundaries are correlated with traditional measures like
\textit{precision} and \textit{recall}. However, these correlations are not linear. For the constant cost model with 1-to-1
relationship between artifacts and defects and assuming perfect quality
assurance, the lower boundary is actually the inverse of \textit{precision}. The
\textit{recall} has no direct relationship with the cost boundaries.
Due to the non-linearity of the relationship, an increase of, e.g., 10\%
\textit{recall} does not mean that the costs are reduced by 10\% or the 
boundaries change by 10\%. For example, Figure~\ref{fig:rep-results}(a) shows
that if the recall for falcon is increased from 80\% to 90\%, there is a
superlinear change of the upper boundary and a roughly linear change of the
lower boundary. Thus, using such metrics for the comparison of models is - to
some degree - a viable solution for the comparison of defect prediction models,
even though the impact of the difference in performance metrics on costs may be
misleading. However, our results also clearly show, that if only such metrics
are considered, the crucial aspect of whether a prediction model can actually
save costs is neglected.

\subsection{Criteria for successful defect prediction}
\label{sec:thresholds}

We consider defect prediction as successful if it can save costs. Our
results show that there is no definitive value for metrics like
\textit{precision} and \textit{recall} where defect prediction can be cost
saving and that this is project dependent. For many projects we considered,
a precision of less than 25\% was sufficient for cost savings. Thus, criteria
that define the success defect prediction using performance metrics like
\textit{recall}, \textit{precision}, and \textit{accuracy} are misleading (e.g.,
\cite{Zimmermann2009, He2012, Herbold2017b}). Using the boundary conditions of
our cost model gives a hard criterion that is required for defect prediction to
be successful, i.e., that the lower boundary must be less than the upper
boundary.

\subsection{The impact of the size}
The trends for the size aware and constant cost model are almost the
same, the difference is only the value for the cost ratio $C$, which is
roughly shifted by multiplying the mean \ac{LOC} of each project. This change
is expected, as $C$ models the relation between quality assurance costs per
complete artifacts and defects for the constant cost models, and the relation
between the quality assurance costs per lines of code and defects for the size
aware cost model. From the literature, we would have assumed that the
differences between the size aware and the constant cost model are larger, e.g,
because of the work by Rahman\etal~\cite{Rahman2012}. However, our results
indicate that while the size has an effect, the overall trend of how the
boundaries behave is the same for the constant cost model and the size aware
cost model. However, the reason for this lack of a stronger effect may be due to
our randomized simulation of defect prediction models. Since size is often a
strong predictor in defect prediction models, sometimes even outperforming
machine learning~\cite{Zhou2018}, the results may change if prediction models that
consider the size during predictions, are used. 

\subsection{Boundary conditions are required but not sufficient}
\label{sec:impact-exec-costs}
While our boundary conditions are required for a
positive profit they are not sufficient. Thus, even if the cost ratio of
quality assurance efforts and defects of a project is within the interval of
cost-saving cost ratios for a defect prediction model as defined by
Corollary~\ref{thm:boundaries-init}, it is possible that there is no positive
profit. The reason for this is our assumption that the one-time and
continuous costs are zero for our initializations. While we believe that these
costs are relatively small (see Section~\ref{sec:one-time-costs}), the interval
between the upper and lower boundary decreases, as is shown by the
corrolaries~\ref{thm:boundaries-lowerboundary} and
\ref{thm:boundaries-upperboundary}. Thus, if the actual cost ratio is very close
to the boundaries, the actual profit could still be negative. 

Additionally, our approach only considers success from the management
perspective, i.e., related only to monetary issues. This does not cover whether software
developers would actually use the defect prediction model as expected and what
would be required to achieve this, as such considerations are out of scope of
this article.

\subsection{Using the cost model}
\label{sec:using-the-cost-model}

Our work demonstrates the need for cost modeling of defect prediction models
that goes beyond the standard measurements of prediction performance like
\textit{precision} and \textit{recall}, because these measures are not directly
related to the actual cost-saving potential of a defect prediction model. For
the adoption of our cost model for future research, we propose the following.

\textbf{1) Adopt the $n$-to-$m$ relationships.} Our findings clearly show that
the results of the realistic $n$-to-$m$ relationships between artifacts and defects
deviate from the simple 1-to-1 and 1-to-n. Therefore either the
size-aware or constant cost $n$-to-$m$ model should be used.

\textbf{2) Evaluate if the prediction model is cost saving.} All papers should
evaluate if the lower bound is really less than the upper bound. If this is not
the case, the defect prediction cannot save costs in comparison to one of the
trivial baselines of either predicting nothing, or predicting everything. 

\textbf{3) Compare the upper and lower boundaries.} In case different defect
prediction models are compared to each other, the upper and lower boundaries are valuable
metrics. We suggest that the boundaries are used instead of \textit{precision}
and \textit{recall} for the comparison of models, as they give a more realistic
picture on the actual performance of the prediction model in a realistic
setting, because the relationship of the costs with \textit{precision} and \textit{recall} 
is not linear~\ref{sec:confusion-matrix-metrics}.
Additionally, project managers can use the upper and lower boundaries
to estimate if the defect prediction model can be cost saving in their
project. Based on their intuition of what defects costs for the
project, they can evaluate if the project/organization is within
the cost saving area of the defect prediction model. The advantage of the
boundaries is that project managers do not need to have exact estimates for the
costs. If they can estimate a range in which the costs are, this is sufficient
to evaluate if the defect prediction model can be cost saving. 

\textbf{4) Compare the range of cost saving ratios $C$.} In addition to the
comparison of the boundaries, the difference between the upper and the lower
boundary should also be considered. We suggest that this comparison replaces
performance measures like the $F-Measure$ as a large range of cost efficient
values indicates that the model performs well under different circumstances.
Moreover, the farther the actual cost ratio of a project is away from the
boundaries, the higher the profit. Thus, defect
prediction models with larger ranges between the boundaries are not only
cost-savings for more projects, but can also save more costs. 

\textbf{5) Evaluate for perfect and imperfect quality assurance.} Our results
show that the cost efficiency of defect prediction models changes with the
effectiveness of the applied quality assurance to reveal predicted defects. We
suggest that both perfect (i.e., $p_{qf}=0$) as well as imperfect (e.g.
$p_{qf}=0.5$) quality assurance is considered for evaluations.

\subsection{Possible Improvements}

While we believe that our general cost model covers the most important factors
and the the initializions are based on reasonable assumptions, there may be
opportunities to further improve the model. 

On the one hand, the cost model could be initialized differently. In this case,
Theorem~\ref{thm:boundaries-random} and
corollaries~\ref{thm:boundaries-lowerboundary} and
\ref{thm:boundaries-upperboundary} would not be affected. Then researchers could
take pattern from Corrolory~\ref{thm:boundaries-init} and calculate the cost
boundaries for the other initializations. For example, the costs for the quality
assurance could be estimated in relation to the complexity of artifacts instead
of the size, or even a combination of both. This way, the cost estimation could
possibly better account for the impact of developer experience on the quality
assurance costs.

On the other hand, there are several ways in which the general cost model may be
extended. One possible extension would be to replace the constants $C_{EXEC}$ for continuous
costs, $C_{QA}$ for the quality assurance costs per unit, and $C_{DEF}$ for the
costs per defect with random variables. All three constants represent the mean
costs that would occur. Consequently, these constants ignore the uncertainty in
the probability distribution of these costs. A more realistic model would be to
replace these constants with random variables that represent the probability
distribution of these costs. This way, the cost model could account for
randomly occurring continuous costs like the change of team members, respect that
the quality assurance costs may vary per unit, and incorporate a more realistic
model for the costs of defects. Mathematically, the current constants would be
the expected value of the random variables, and, consequently, the cost
boundaries we calculated would represent the expected cost boundaries. The
random variables would enable a mathematical analysis of the uncertainty of
these cost boundaries. However, a necessary precursor is that the probability
distributions of these costs would be known, as the uncertainty of the costs
could not be expressed otherwise. 

Another possible extension is to incorporate additional cost factors into the
cost model, e.g., drawbacks and benefits of the approach that are not directly
associated to costs. For example, prediction models with a low
\textit{precision} may still be cost efficient, but they could potentially also
be frustrating for developers due to the large amount of false positives. On the
other hand, the developers would gain more experience with the artifacts that
are false positives and they may also develop more (automated) tests for these
artifacts. Thus, there would be an indirect gain in experience and possible
future cost savings due to the larger test suite. Such cost terms could be added
to the general model in Equation~\eqref{eq:cost1}. As a consequence,
Theorem~\ref{thm:boundaries-random} would have to account for these cost terms,
which may modify the cost boundaries. A similar approach would be to subdivide
the existing cost factors. E.g., the quality assurance costs could be divided
into the costs due a potentially delayed release and the costs for additional
man power required for the quality assurance. Alternatively, the cost model
could be used as is, but instantiated twice: once with only man power
considerations, once with only time considerations. This would allow managers to
evaluate if the prediction model may be more expensive in terms of man power,
but effective in terms of loss due to a delayed release.

A further possible extension of the cost model is to subdivide the cost
factors, e.g., to subdivide the costs of defects into costs for security issues,
costs for defects that crash the application, and costs for other defects.
However, this extension is incompatible  with the mathematical analysis we
conducted in this article. We were only able to proof the boundary conductions
because we could reduce the number of unknown variables to one by considering
the costs for quality assurance and the costs for defects as a ratio. If we
would, e.g., subdivide the costs of defects into three subcategories, we would
have three such ratios as unknown variables and would need to establish boundary
conditions for these variables that would not be independent of each other. Thus, we
do not believe that such an extension of this cost model is feasible.\footnote{We note
that extensions with additional cost terms may lead to the similar problems with
the mathemical analysis.} Regardless, if random variables are used
as we described at the beginning of this section, the distribution of the random
variable could account for the different types of defects, e.g., by modelling
the distribution of the cost of defects as a mixture of the distributions of the
costs of security issues, costs of crashes, and costs of other defects.

\subsection{Threats to Validity}
There are several threats to the validity of our work, which we report following
the classification by Wohlin\etal~\cite{Wohlin2012}.

\subsubsection{Construct Validity}
The evaluation of the boundary conditions through simulation of defect
predictions may be unsuited for the evaluation of trends. We mitigated this by
simulating the prediction an real-world data and sampling accross a large range
of prediction performances. There may be a defect in the code for the
simulations of defect prediction and the evaluation of them using our cost model
for the experiments performed in Section~\ref{sec:experiments}. However, the
source code is relatively short and not very complex. Moreover, we reviewed the
source code to minimize this possibility.

\subsubsection{Internal Validity}
\label{sec:internal-threats}

There may be important factors influencing the costs of defect predictions,
which we did not include in the general model. We scanned the literature
regarding related work to costs of defect predictions to mitigate this threat.
Moreover, we initialized the cost model
using different assumptions. These assumptions may be unrealistic or wrong,
leading to wrong conclusions. We explained the rationale behind each
design decision and, if possible, grounded them in prior work from the
literature to mitigate this threat.

We have only presented the results of the trends of the cost model with respect
to \textit{precision} and \textit{recall}, because these are the most common
metrics used for defect prediction research. These trends may look different
using other metrics. However, other common metrics are either directly or
indirectly related to \textit{precision} and \textit{recall}, e.g., the
\textit{F-Measure}, \textit{G-Measure}, or \textit{AUC}. This mitigates the
threat that our conclusions, especially regarding the impact of the $n$-to-$m$
relationship may be wrong.

The collected data may also contain problems we have not addressed causing noise
in the data. Commits may reference multiple issues, which could lead to double
counting of files for defects. However, only 46 of the 1493 commits that address
defects reference multiple issues, i.e., the effect of this would be very small. 

Moreover, we have not manually validated that all changes within a commit that
fixes a defect are really part of the correction of the defect. This means that
our data may have an inflated number of files per defect, which could bias the
results towards showing differences between the $n$-to-$m$ cost model and the
1-to-1 and the 1-to-$m$ cost model. To mitigate this threat, we cross-checked our ratio of files per
defect $mean(|d|)$ with the results by Mills\etal~\cite{Mills2018}, who manually
validated the file actions. Based on the results by Mills\etal,
$mean(|d|)$ falls into the interval $[1.11, 2.13]$ with 99.5\% confidence. Our
data has a mean value of 1.56. Thus, even if there are false positive file
actions in our data, the results from the experiment should still be
representative. 

Additionally, we took a closer look at the projects tika and zookeeper, as
Mills\etal~\cite{Mills2018} manually validated data from these projects,
although from a different time period than our data. We evaluated how many
changes to files Mills\etal{} manually determined as part of the
correction of a defect and compared this to the number of of changes to files
that we identify with our heuristic for the removal of false positive file
changes based on changes to whitespaces, comments, and tests. For tika, Mills et
al. identified 22 changes to files for the correction of defects and 27 changes
to files that were unrelated to defects. We correctly filtered 22 of the 27
unrelated changes. For zookeeper, Mills\etal{} identified 114 changes to files
for the correction of defects and 126 unrelated changes. We correctly filtered
98 of the 126 unrelated changes. Overall, we filtered 78\% of the unrelated
changes. Thus, we removed most of the noise from our data. Consequently, it is highly
unlikely that the remaining noise is strong enough to alter our results to such
a degree that no differences due to the $n$-to-$m$ relationship would be
visible anymore. 

\subsubsection{External Validity}

While we evaluated our cost model on fifteen real-world projects, we only used
simulation and did not use actual defect prediction models. From our results, we
extrapolate that our cost model is required and can affect the results of
evaluations for all real-world project where a subset of the defects affect
multiple files, due to the $n$-to-$m$ relationship. However, we cannot
definitively conclude this.

\subsubsection{Reliability}

The filtering of issues whether they are really defects or not may affect the
results of this article and depends on the author. To mitigate this threat to
the reliability, we followed the same rules for defects as
Herzig\etal~\cite{Herzig2013} and documented all decisions in the replication
kit. Due to the results of Herzig\etal~\cite{Herzig2013}, we expected to discard between
27.4\% and 42.9\% of the issues with 99.5\% confidence. Thus, the
$\frac{175}{413}=42.4\%$ discarded issues are within the bounds established by
the state of the art, indicating that this study is reliable.

%% file: conclusion.tex
\section{Conclusion}
\label{sec:conclusion}

In this article, we specified a cost model for software defect prediction,
showed how the cost model can be used to calculate the profit of defect
prediction, and defined mathematically provable boundaries that defect prediction
must fulfill in order to allow for a positive profit under any circumstances.
We have shown how our cost model can be initialized using different assumptions.
Using simulated defect prediction data, we have analyzed the impact of the
assumptions on the costs.
Using these insights, we provide guidelines for using our cost model in future
research.  Moreover, we discovered a flaw in all current defect prediction data
sets and consequently also all evaluations of defect prediction approaches due
to an oversimplification of the relationship between software artifacts and
defects.

In future work, we will apply our cost model to the state of the art of defect
prediction and assess under which conditions the predictions are successful
and compare defect prediction models with respect to their cost-saving
potential. However, before such an analysis is possible, we will work on
creating a new defect prediction data set that allows for evaluations that
respect the $n$-to-$m$ relationship between software artifacts and defects as our
simulations show that cost estimation are very different if this aspect is not
considered. Another important aspect of future work are measures for the
acceptance of defect prediction models by develepors. While the profit is an
imporant indicator for success of a technique from the management perspective,
tools that apply defect prediction must be used by developers. For example, our results show that a
positive profit can sometimes be achieved even with a low \textit{precision},
i.e., many false positives. However, whether developers would accept this and
under which circumstances they may accept a high number of false positives has
not yet been sufficiently addressed in the literature.

%% file: main.bbl
\begin{thebibliography}{10}
\providecommand{\url}[1]{#1}
\csname url@samestyle\endcsname
\providecommand{\newblock}{\relax}
\providecommand{\bibinfo}[2]{#2}
\providecommand{\BIBentrySTDinterwordspacing}{\spaceskip=0pt\relax}
\providecommand{\BIBentryALTinterwordstretchfactor}{4}
\providecommand{\BIBentryALTinterwordspacing}{\spaceskip=\fontdimen2\font plus
\BIBentryALTinterwordstretchfactor\fontdimen3\font minus
  \fontdimen4\font\relax}
\providecommand{\BIBforeignlanguage}[2]{{%
\expandafter\ifx\csname l@#1\endcsname\relax
\typeout{** WARNING: IEEEtran.bst: No hyphenation pattern has been}%
\typeout{** loaded for the language `#1'. Using the pattern for}%
\typeout{** the default language instead.}%
\else
\language=\csname l@#1\endcsname
\fi
#2}}
\providecommand{\BIBdecl}{\relax}
\BIBdecl

\bibitem{Catal2009a}
C.~Catal and B.~Diri, ``A systematic review of software fault prediction
  studies,'' \emph{Expert Systems with Applications}, vol.~36, no.~4, pp. 7346
  -- 7354, 2009.

\bibitem{Hall2012}
T.~Hall, S.~Beecham, D.~Bowes, D.~Gray, and S.~Counsell, ``A systematic
  literature review on fault prediction performance in software engineering,''
  \emph{IEEE Trans. Softw. Eng.}, vol.~38, no.~6, pp. 1276--1304, Nov. 2012.

\bibitem{Hosseini2017a}
S.~Hosseini, B.~Turhan, and D.~Gunarathna, ``A systematic literature review and
  meta-analysis on cross project defect prediction,'' \emph{IEEE Trans. Softw.
  Eng.}, p.~1, 2017.

\bibitem{Herbold2017b}
S.~Herbold, A.~Trautsch, and J.~Grabowski, ``A comparative study to benchmark
  cross-project defect prediction approaches,'' \emph{IEEE Trans. Softw. Eng.},
  vol.~PP, no.~99, pp. 1--1, 2017.

\bibitem{Nam2017}
J.~Nam, W.~Fu, S.~Kim, T.~Menzies, and L.~Tan, ``Heterogeneous defect
  prediction,'' \emph{IEEE Trans. Softw. Eng.}, vol.~PP, no.~99, pp. 1--1,
  2017.

\bibitem{Jing2015}
X.~Jing, F.~Wu, X.~Dong, F.~Qi, and B.~Xu, ``Heterogeneous cross-company defect
  prediction by unified metric representation and cca-based transfer
  learning,'' in \emph{Proc. of the 2015 10th Joint Meeting on Foundations of
  Software Engineering}, ser. ESEC/FSE 2015.\hskip 1em plus 0.5em minus
  0.4em\relax New York, NY, USA: ACM, 2015, pp. 496--507.

\bibitem{Nam2015b}
J.~Nam and S.~Kim, ``Clami: Defect prediction on unlabeled datasets,'' in
  \emph{Automated Software Engineering (ASE), 2015 30th IEEE/ACM Int. Conf.
  on}, Nov 2015, pp. 452--463.

\bibitem{Zhang2016}
F.~Zhang, Q.~Zheng, Y.~Zou, and A.~E. Hassan, ``Cross-project defect prediction
  using a connectivity-based unsupervised classifier,'' in \emph{Proc. of the
  38th Int. Conf. on Software Engineering}.\hskip 1em plus 0.5em minus
  0.4em\relax ACM, 2016.

\bibitem{Huang2017}
Q.~Huang, X.~Xia, and D.~Lo, ``Supervised vs unsupervised models: A holistic
  look at effort-aware just-in-time defect prediction,'' in \emph{2017 IEEE
  Int. Conf. on Software Maintenance and Evolution (ICSME)}, Sept 2017, pp.
  159--170.

\bibitem{Tantithamthavorn2017}
C.~Tantithamthavorn, S.~McIntosh, A.~E. Hassan, and K.~a. Matsumoto, ``An
  empirical comparison of model validation techniques for defect prediction
  models,'' \emph{IEEE Trans. Softw. Eng.}, vol.~43, no.~1, pp. 1--18, 2017.

\bibitem{Tantithamthavorn2016}
C.~Tantithamthavorn, S.~McIntosh, A.~E. Hassan, and K.~Matsumoto, ``Automated
  parameter optimization of classification techniques for defect prediction
  models,'' in \emph{Proc. of the 38th Int. Conf. on Software
  Engineering}.\hskip 1em plus 0.5em minus 0.4em\relax ACM, 2016.

\bibitem{Krishna2018}
R.~Krishna and T.~Menzies, ``Bellwethers: A baseline method for transfer
  learning,'' \emph{IEEE Trans. Softw. Eng.}, pp. 1--1, 2018.

\bibitem{Tantithamthavorn2018}
C.~Tantithamthavorn and A.~E. Hassan, ``An experience report on defect
  modelling in practice: Pitfalls and challenges,'' \emph{The Int. Conf. on
  Software Engineering: Software Engineering in Practice Track (ICSE-SEIP)},
  vol.~15, no.~17, pp. 71--73, 2018.

\bibitem{Yang2016}
Y.~Yang, Y.~Zhou, J.~Liu, Y.~Zhao, H.~Lu, L.~Xu, B.~Xu, and H.~Leung,
  ``Effort-aware just-in-time defect prediction: Simple unsupervised models
  could be better than supervised models,'' in \emph{Proc. of the 2016 24th ACM
  SIGSOFT Int. Symp. on Foundations of Software Engineering}.\hskip 1em plus
  0.5em minus 0.4em\relax ACM, 2016.

\bibitem{Zhou2018}
Y.~Zhou, Y.~Yang, H.~Lu, L.~Chen, Y.~Li, Y.~Zhao, J.~Qian, and B.~Xu, ``How far
  we have progressed in the journey? an examination of cross-project defect
  prediction,'' \emph{ACM Trans. Softw. Eng. Methodol.}, vol.~27, no.~1, pp.
  1:1--1:51, Apr. 2018.

\bibitem{Herbold2017}
\BIBentryALTinterwordspacing
S.~Herbold, ``A systematic mapping study on cross-project defect prediction,''
  \emph{CoRR}, vol. abs/1705.06429, 2017. [Online]. Available:
  \url{https://arxiv.org/abs/1705.06429}
\BIBentrySTDinterwordspacing

\bibitem{Ohlsson1996}
N.~Ohlsson and H.~Alberg, ``Predicting fault-prone software modules in
  telephone switches,'' \emph{IEEE Trans. Softw. Eng.}, vol.~22, no.~12, pp.
  886--894, Dec. 1996.

\bibitem{Rahman2012}
F.~Rahman, D.~Posnett, and P.~Devanbu, ``Recalling the ``imprecision'' of
  cross-project defect prediction,'' in \emph{Proc. ACM SIGSOFT 20th Int. Symp.
  Found. Softw. Eng. (FSE)}.\hskip 1em plus 0.5em minus 0.4em\relax ACM, 2012.

\bibitem{Jiang2008}
Y.~Jiang, B.~Cukic, and Y.~Ma, ``\BIBforeignlanguage{English}{Techniques for
  evaluating fault prediction models},''
  \emph{\BIBforeignlanguage{English}{Empirical Softw. Eng.}}, vol.~13, no.~5,
  pp. 561--595, 2008.

\bibitem{Hemmati2015}
H.~Hemmati, M.~Nagappan, and A.~E. Hassan, ``Investigating the effect of
  "defect co-fix" on quality assurance resource allocation: A search-based
  approach,'' \emph{Journal of Systems and Software}, vol. 103, pp. 412 -- 422,
  2015.

\bibitem{Arisholm2006}
E.~Arisholm and L.~C. Briand, ``Predicting fault-prone components in a {Java}
  legacy system,'' in \emph{Proc. 5th ACM/IEEE Int. Symp. Emp. Softw. Eng.
  (ISESE)}.\hskip 1em plus 0.5em minus 0.4em\relax ACM, 2006.

\bibitem{Canfora2013}
G.~Canfora, A.~D. Lucia, M.~D. Penta, R.~Oliveto, A.~Panichella, and
  S.~Panichella, ``Multi-objective cross-project defect prediction,'' in
  \emph{Proc. 6th IEEE Int. Conf. Softw. Testing, Verification and Validation
  (ICST)}, 2013.

\bibitem{Jureczko2010}
M.~Jureczko and L.~Madeyski, ``Towards identifying software project clusters
  with regard to defect prediction,'' in \emph{Proc. 6th Int. Conf. on
  Predictive Models in Softw. Eng. (PROMISE)}.\hskip 1em plus 0.5em minus
  0.4em\relax ACM, 2010.

\bibitem{Zhang2015a}
Y.~Zhang, D.~Lo, X.~Xia, and J.~Sun, ``An empirical study of classifier
  combination for cross-project defect prediction,'' in \emph{Computer Software
  and Applications Conf. (COMPSAC), 2015 IEEE 39th Annual}, vol.~2, July 2015,
  pp. 264--269.

\bibitem{Khoshgoftaar1998}
T.~M. Khoshgoftaar and E.~B. Allen, ``Classification of fault-prone software
  modules: Prior probabilities, costs, and model evaluation,'' \emph{Emp.
  Softw. Eng.}, vol.~3, no.~3, pp. 275--298, Sep 1998.

\bibitem{Khoshgoftaar2004}
T.~M. Khoshgoftaar and N.~Seliya, ``Comparative assessment of software quality
  classification techniques: An empirical case study,'' \emph{Emp. Softw.
  Eng.}, vol.~9, no.~3, pp. 229--257, Sep 2004.

\bibitem{Drummond2006}
C.~Drummond and R.~Holte, ``\BIBforeignlanguage{English}{Cost curves: An
  improved method for visualizing classifier performance},''
  \emph{\BIBforeignlanguage{English}{Machine Learning}}, vol.~65, no.~1, pp.
  95--130, 2006.

\bibitem{Zhang2013}
H.~Zhang and S.~C. Cheung, ``A cost-effectiveness criterion for applying
  software defect prediction models,'' in \emph{Proc. of the 2013 9th Joint
  Meeting on Foundations of Software Engineering}, ser. ESEC/FSE 2013.\hskip
  1em plus 0.5em minus 0.4em\relax New York, NY, USA: ACM, 2013, pp. 643--646.

\bibitem{Pham2003}
H.~Pham, ``{Software reliability and cost models: Perspectives, comparison, and
  practice},'' \emph{European Journal of Operational Research}, vol. 149,
  no.~3, pp. 475 -- 489, 2003.

\bibitem{Kingman1992}
J.~Kingman, \emph{Poisson Processes}.\hskip 1em plus 0.5em minus 0.4em\relax
  Clarendon Press, 1992, vol.~3.

\bibitem{Stolfo2000}
S.~J. {Stolfo}, {Wei Fan}, {Wenke Lee}, A.~{Prodromidis}, and P.~K. {Chan},
  ``Cost-based modeling for fraud and intrusion detection: results from the jam
  project,'' in \emph{Proc. DARPA Information Survivability Conf. and
  Exposition. DISCEX'00}, vol.~2, Jan 2000, pp. 130--144 vol.2.

\bibitem{Patry2015}
G.~Patry, A.~Romagny, S.~Martinet, and D.~Froelich, ``Cost modeling of
  lithium-ion battery cells for automotive applications,'' \emph{Energy Science
  \& Engineering}, vol.~3, no.~1, pp. 71--82, 2015.

\bibitem{Etkin2004}
D.~S. Etkin, ``Modeling oil spill response and damage costs,'' in \emph{Proc.
  of the Fifth Biennial Freshwater Spills Symp.}, 2004.

\bibitem{Pugliatti2007}
M.~Pugliatti, E.~Beghi, L.~Forsgren, M.~Ekman, and P.~Sobocki, ``Estimating the
  cost of epilepsy in europe: A review with economic modeling,''
  \emph{Epilepsia}, vol.~48, no.~12, pp. 2224--2233, 2007.

\bibitem{Nakamura2010}
E.~Nakamura and J.~Steinsson, ``{Monetary Non-neutrality in a Multisector Menu
  Cost Model*},'' \emph{The Quarterly Journal of Economics}, vol. 125, no.~3,
  pp. 961--1013, 08 2010.

\bibitem{Canfora2015}
G.~Canfora, A.~D. Lucia, M.~D. Penta, R.~Oliveto, A.~Panichella, and
  S.~Panichella, ``Defect prediction as a multiobjective optimization
  problem,'' \emph{Software Testing, Verification and Reliability}, vol.~25,
  no.~4, pp. 426--459, 2015.

\bibitem{DAmbros2012}
M.~D'Ambros, M.~Lanza, and R.~Robbes, ``Evaluating defect prediction
  approaches: A benchmark and an extensive comparison,'' \emph{Empirical Softw.
  Engg.}, vol.~17, no. 4-5, pp. 531--577, Aug. 2012.

\bibitem{Hosseini2017}
S.~Hosseini, B.~Turhan, and M.~Mäntylä, ``A benchmark study on the
  effectiveness of search-based data selection and feature selection for cross
  project defect prediction,'' \emph{Information and Software Technology},
  2017.

\bibitem{Trautsch2017}
F.~Trautsch, S.~Herbold, P.~Makedonski, and J.~Grabowski, ``Addressing problems
  with replicability and validity of repository mining studies through a smart
  data platform,'' \emph{Emp. Softw. Eng.}, Aug 2017.

\bibitem{replication-kit}
\BIBentryALTinterwordspacing
S.~Herbold, ``sherbold/replication-kit-tse-2018-costmodel: Release of the
  replication kit,'' Nov 2019. [Online]. Available:
  \url{https://doi.org/10.5281/zenodo.3537703}
\BIBentrySTDinterwordspacing

\bibitem{Herzig2013}
K.~Herzig, S.~Just, A.~Rau, and A.~Zeller, ``Predicting defects using change
  genealogies,'' in \emph{Software Reliability Engineering (ISSRE), 2013 IEEE
  24th Int. Symp. on}, Nov 2013, pp. 118--127.

\bibitem{DAmbros2010}
M.~D'Ambros, M.~Lanza, and R.~Robbes, ``{An Extensive Comparison of Bug
  Prediction Approaches},'' in \emph{Proc. of the 7th IEEE Working Conf. on
  Mining Software Repositories (MSR)}.\hskip 1em plus 0.5em minus 0.4em\relax
  IEEE Computer Society, 2010.

\bibitem{Wu2011}
R.~Wu, H.~Zhang, S.~Kim, and S.-C. Cheung, ``Relink: Recovering links between
  bugs and changes,'' in \emph{Proc. of the 19th ACM SIGSOFT Symp. and the 13th
  European Conf. on Foundations of Software Engineering}, ser. ESEC/FSE
  '11.\hskip 1em plus 0.5em minus 0.4em\relax New York, NY, USA: ACM, 2011, pp.
  15--25.

\bibitem{Zimmermann2009}
T.~Zimmermann, N.~Nagappan, H.~Gall, E.~Giger, and B.~Murphy, ``Cross-project
  defect prediction: a large scale experiment on data vs. domain vs. process,''
  in \emph{Proc. the 7th Joint Meet. Eur. Softw. Eng. Conf. (ESEC) and the ACM
  SIGSOFT Symp. Found. Softw. Eng. (FSE)}.\hskip 1em plus 0.5em minus
  0.4em\relax ACM, 2009, pp. 91--100.

\bibitem{He2012}
Z.~He, F.~Shu, Y.~Yang, M.~Li, and Q.~Wang, ``\BIBforeignlanguage{English}{An
  investigation on the feasibility of cross-project defect prediction},''
  \emph{\BIBforeignlanguage{English}{Automated Softw. Eng.}}, vol.~19, pp.
  167--199, 2012.

\bibitem{Wohlin2012}
C.~Wohlin, P.~Runeson, M.~H\"{o}st, M.~C. Ohlsson, B.~Regnell, and A.~Wesslen,
  \emph{{Experimentation in Software Engineering}}.\hskip 1em plus 0.5em minus
  0.4em\relax Springer Publishing Company, Incorporated, 2012.

\bibitem{Mills2018}
C.~{Mills}, J.~{Pantiuchina}, E.~{Parra}, G.~{Bavota}, and S.~{Haiduc}, ``Are
  bug reports enough for text retrieval-based bug localization?'' in \emph{2018
  IEEE Int. Conf. on Software Maintenance and Evolution (ICSME)}, Sep. 2018,
  pp. 381--392.

\end{thebibliography}
